\def\format{plain}
\def\plain{plain}
\ifx \format \plain
  \documentclass[11pt]{article}
  \usepackage[margin=1in,footskip=0.25in]{geometry}
\else
  \documentclass{llncs}
\fi
\usepackage{amsmath}
\usepackage[utf8]{inputenc}
\usepackage{physics}
\usepackage{xcolor}
\usepackage{hyperref}
\usepackage{cleveref}

\usepackage{amsthm}

\usepackage[]{cryptocode}
\usepackage{empheq}
\ifx \format \plain
  \newtheorem{theorem}{Theorem}
\else

\fi
\newtheorem{proposition}[theorem]{Proposition}

\newtheorem{lemma}[theorem]{Lemma}

\newtheorem{corollary}[theorem]{Corollary}

\newtheorem{definition}[theorem]{Definition}

\newtheorem{rem}{Remark}
\newtheorem{algorithm}{Algorithm}

\crefname{lemma}{lemma}{lemmas}
\Crefname{lemma}{Lemma}{Lemmas}

\usepackage{qcircuit}
\usepackage{mathrsfs}
\usepackage{stmaryrd}
\usepackage{amssymb}
\usepackage[normalem]{ulem}

\newcommand{\cO}{\mathsf{cO}}

\newcommand{\new}[1]{#1}

\newcommand{\stateiafter}{\psi_i}
\newcommand{\stateibefore}{\phi_i}
\newcommand{\stateiminusoneafter}{\psi_{i-1}}

\newcommand{\colh}[2]{{\stackrel{}{\widetilde{f}^{col}_{#1,#2}}}}
\newcommand{\acolh}[2]{{f_{#1,#2}^{col}}}
\newcommand{\tworsc}[1]{{g_{#1}}}

\newcommand{\ktworsclcol}[3]{{\widehat{g}_{#1,#2,#3}}}
\newcommand{\ktworsc}[2]{{g_{#1,#2}}}

\newcommand{\sminusonersc}[2]{{f_{#1,#2}}}
\newcommand{\ssrsc}[3]{{\widehat{g}_{#1,#2,#3}}}
\newcommand{\srsc}[2]{{g_{#1,#2}}}

\newcommand{\aljrepet}[3]{{\stackrel{}{\widetilde{f}^{rep}_{#1,#2,#3}}}}
\newcommand{\ljrepet}[3]{{{f}_{#1,#2,#3}^{rep}}}
\newcommand{\ksc}[2]{{g_{#1,#2}}}

\title{Quantum security of subset cover problems\footnote{This work was partially funded by 
    PEPR integrated project EPiQ ANR-22-PETQ-0007 part of Plan France 2030.}}
\ifx \format \plain
  \author{Samuel Bouaziz--Ermann\footnote{samuel.bouaziz@ens-rennes.fr}, Alex B. Grilo\footnote{Alex.Bredariol-Grilo@lip6.fr}~ and Damien Vergnaud~\footnote{damien.vergnaud@lip6.fr} \\LIP6, Sorbonne Université, CNRS}
  \date{}
\else
  \author{}
  \institute{}
\fi

\begin{document}

\maketitle

\vspace{-3em}

\begin{abstract}
  The subset cover problem for $k \geq 1$ hash functions, which can be seen as an extension of the collision problem, was introduced in 2002 by Reyzin and Reyzin to analyse the security of their hash-function based signature scheme HORS. 
  The security of many hash-based signature schemes relies on this problem or a variant of this problem (e.g. HORS, SPHINCS, SPHINCS+, \dots). 
  
  Recently, Yuan, Tibouchi and Abe (2022) introduced a variant to the subset cover problem, called restricted subset cover, and proposed a quantum algorithm for this problem. In this work,  we prove that any quantum algorithm needs to make $\Omega\left((k+1)^{-\frac{2^{k}}{2^{k+1}-1}}\cdot N^{\frac{2^{k}-1}{2^{k+1}-1}}\right)$ queries to the underlying hash functions with codomain size $N$ to solve the restricted subset cover problem, which essentially matches the query complexity of the algorithm proposed by Yuan, Tibouchi and Abe. 
  
  We also analyze the security of the general $(r,k)$\---subset cover problem, which is the underlying problem that implies the unforgeability of HORS under a $r$-chosen message attack (for $r \geq 1$). We prove that a generic quantum algorithm needs to make $\Omega\left(N^{k/5}\right)$ queries to the underlying hash functions to find a $(1,k)$\---subset cover. 
  We also propose a quantum algorithm that finds a $(r,k)$\---subset cover making $O\left(N^{k/(2+2r)}\right)$ queries to the $k$ hash functions.
\end{abstract}

\section{Introduction}
Cryptographic hash functions are functions mapping arbitrary-length inputs to fixed-length outputs and are one of the central primitives in cryptography.  They serve as building blocks for numerous cryptographic primitives such as key-establishment,  authentication,  encryption, or digital signatures.  In particular, \emph{one-time signatures} -- i.e.\ in which the signing key can
be used only once -- based only on hash functions were proposed by Lamport as soon as in 1979~\cite{Lamport}.  The basic idea is to evaluate a cryptographic hash function on secret values to generate the public verification key and to authenticate a single message by revealing a subset of those secret pre-images. 

With the development of quantum technologies, which may bring drastic attacks against widely
deployed cryptographic schemes based on the hardness of integer factorization or the discrete logarithm~\cite{Shor},  hash-based signatures have regained interest within the realm of "post-quantum" cryptography and the recent NIST standardization process.  In particular, the SPHINCS+ candidate~\cite{CCS:BHKNRS19} has been selected in 2022 for standardization by NIST and other constructions  are standardized by IETF/IRTF.  The SPHINCS+ signature scheme and its predecessor SPHINCS~\cite{EC:BHHLNP15} make use of a Merkle-hash tree and of HORST, a variant of a hash-based scheme called HORS~\cite{ACISP:ReyRey02}.  HORS (for “Hash to Obtain Random Subset”) uses a hash function to select the subset of secret pre-images to reveal in a signature and the knowledge of these secrets for several subsets may not be enough to produce a forgery,  a property that makes HORS a \emph{few-time signature} scheme.

More concretely, the security of HORS (and HORST) relies on the hardness of finding a subset cover (SC) for the underlying hash function.  More formally, to define the \emph{$(r,k)$\---SC problem},  we consider the hash function as the concatenation of  $k \geq 1$ hash functions $h_1$, \dots, $h_k$ (with smaller outputs) and the problem is to find,  for some integer $r \geq 1$,   $r+1$ elements $x_0,x_1\dots,x_r$ in the hash function domain such that $x_0 \notin \{x_1,\dots,x_r\}$, and
$$
\left\{h_i(x_0) \vert 1\leq i \leq k\right\} \subseteq \bigcup_{j=1}^r \left\{h_i(x_j) \vert 1\leq i \leq k\right\}.
$$
The hardness of this problem for concrete popular hash functions has not been studied in depth but Aumasson and Endignoux~\cite{EPRINT:AumEnd17a} proved in 2017  a lower bound on the number of queries to hash functions 
for the SC problem in the Random Oracle Model (ROM).  However,  the exact security of HORS (and more generally HORST, SPHINCS and SPHINCS+)  with respect to quantum attacks is still not clear. Since quantum computing provides speedups for many problems (e.g.  Grover’s search algorithm~\cite{STOC:Grover96} and Brassard, H\o{}yer, and Tapp \cite{LATIN:BHT98} collision search algorithm),  it is important to provide lower bounds in a quantum world. 

\subsection{Our results}
In this paper, we explore the difficulty of finding subset cover for idealized hash functions for quantum algorithms. We also consider a variant called the $k$-restricted subset cover ($k$\---RSC) problem where, given $k$ functions $h_1,\dots,h_k:\mathcal{X} \rightarrow \mathcal{Y}$ such that $N=|\mathcal{Y}|$, one has to find $k+1$ elements $x_0,x_1\dots,x_k$ such that:
$$
\forall 1\leq i\leq k, h_i(x_0)= h_i(x_i)
$$
and $x_0 \notin \{x_1,\dots,x_k\}$. This variant was defined recently by Yuan, Tibouchi and Abe~\cite{subres}, who showed a quantum algorithm to solve it.
The main contributions of this work are:
\begin{enumerate}
\item \textbf{Lower bound on $k$\---RSC}: we prove that $\Omega\left((k+1)^{-\frac{2^{k}}{2^{k+1}-1}}\cdot N^{\frac{2^{k}-1}{2^{k+1}-1}}\right)$ quantum queries to the idealized hash functions are needed to find a $k$\---RSC with constant probability. \\ 
  (Theorem~\ref{theorem:krsc})

\item \textbf{Lower bound on $(1,k)$\---SC}: we prove that $\Omega\left((k!)^{-1/5} \cdot N^{k/5}\right)$ quantum queries to the idealized hash functions are needed to find a $(1,k)$\---SC with constant probability. \\ %
  (Theorem \ref{theorem:1ksc})

\item \textbf{Upper bound on $(r,k)$\---SC:} we present a quantum algorithm that finds a $(r,k)$\---SC with constant probability with $O\left(N^{k/(2+2r)}\right)$ queries to the hash functions when $k$ is divisible by $r+1$, and $O\left(N^{k/(2+2r) + 1/2}\right)$ otherwise. \\
  (Theorem \ref{theorem:rkscalgo})
\end{enumerate}

\subsection{Technical Overview}
To prove our lower bounds on the query complexity, we use the technique called \emph{compressed random oracle model} introduced by Zhandry in \cite{C:Zhandry19}. Its goal is to record information about the queries of an adversary $A$ in the quantum random oracle model and it permits ``on-the-fly'' simulation of random oracles (or \emph{lazy sampling})
by considering the uniform superposition of all possible random oracles instead of picking a single random oracle at the beginning of the computation. The technique uses a register to keep a record of a so-called \emph{database} of the random oracle and this register is updated whenever $A$ makes a query to the random oracle. At the end of $A$'s computation, the reduction can measure the register of the database, and the distribution of the outputs is uniformly random, as if we had chosen a random oracle at the beginning of its computation. This new register that contains the database is at the gist of our lower bounds.

In Section \ref{section:rsc}, we prove the lower bound on the query complexity to solve the RSC problem. We consider an algorithm $A$ after $i$ quantum queries to the random oracle and call its state at this moment $\ket{\stateiafter}$.
Our goal is to compute an upper bound for the value $|P^{RSC}_k\ket{\stateiafter}|^2$, where $P^{RSC}_k$ is the projection onto the databases that contain a $k$\---RSC.
Computing such a bound leads to a lower bound on the number of queries needed for solving $k$\---RSC with constant probability. To prove our bound, we proceed by induction: assuming we proved a bound for the $k'$\---RSC problem for all $k'<k$, we prove a bound for the $k$\---RSC problem.
The analysis is naturally divided into two parts: whenever $A$ finds a $k$\---RSC after $i$ quantum queries, it means that  either:
\begin{enumerate}
\item $A$ finds it after $i-1$ quantum queries;
\item or $A$ finds it with the $i^{th}$ quantum query.
\end{enumerate}
The first case is recursive and it remains to bound $|P^{RSC}_k\ket{\stateiafter}|$ in the second case. Here, %
the database (after $i-1$ quantum queries) must contain a certain number of $k'$\---RSC (for some $k' < k$), in order for A to find $k$\---RSC with the $i^{th}$ query.
Using this strategy, we obtain a recursive formula from which we can deduce the bound on $|P^{RSC}_k\ket{\stateiafter}|$.

In Section \ref{section:bound1ksc}, we prove a lower bound for the $(1,k)$\---SC problem.
The idea of the proof is similar to the proof for the lower bound of the $k$\---RSC problem but we have to compute a bound for another problem that we define: the \emph{$j$\---repetition} problem.

Finally in Sections \ref{section:algo1ksc} and \ref{section:algorksc}, we design a family of quantum algorithms for finding a $(r,k)$\---SC.
These algorithms are inspired by the algorithm from \cite{subres} to solve the $k$\---RSC problem and \cite{EC:LiuZha19}'s algorithm for finding multi-collisions.  These algorithms are recursive and take as input two parameters $t,k' \in \mathbb{N}$ and perform the following:
\begin{enumerate}
\item Find $t$ distinct $(r-1,k')$\---SC;
\item Find the $(r,k)$\---SC.
\end{enumerate}
The parameters $t$ and $k'$ are chosen in order to optimize the complexity of the algorithm.
The first step is done by applying $r-1$ times the algorithm for the value $k'$, and the second step uses Grover's algorithm.

\subsection{Related works, discussion and open problems}

\newcommand{\mypar}[1]{~\\ \noindent \textbf{#1.} }

\vspace{-1.2em}

\mypar{Collision-finding} The link between finding a multi-collision and finding a subset cover was first discussed in \cite{subres}, since their algorithm is inspired from the one for finding multi-collisions in \cite{EC:LiuZha19}.
In the latter, they also show a lower bound for finding multi-collisions, and our proof of lower bounds uses the same technique they used.
We make use of the compressed oracle technique, first introduced by Zhandry in \cite{C:Zhandry19}, and generalize the proof of the lower bound on multi-collisions to the RSC and SC problems.%

\mypar{Restricted Subset Cover} There is currently only one quantum algorithm for finding RSC \cite{subres}.
Our lower bound for finding a RSC matches their upper bound when $k$, the number of functions, is constant.
However when $k$ is not a constant, their algorithm makes $O\left(k\cdot N^{\frac{2^{k}-1}{2^{k+1}-1}}\right)$ queries to $h_1,\dots,h_k$, which roughly leaves a $k^{3/2}$ gap between the best known attack and our lower bound. To the best of our knowledge, this is the first lower bound on the RSC problem for a quantum algorithm, and there are no such result for classical algorithms.
It would be interesting to see if we can close this gap further.

\mypar{Tighter bounds for $(1,k)$\---SC} 
When $k$ is constant, the lower bound for $(1,k)$\---SC is $\Omega\left(N^{k/5}\right)$, while our algorithm for this problem makes $O\left(N^{k/4}\right)$ queries to the oracle (when $k$ is even).
It would be interesting to tighten this gap, especially since the results for $(1,k)$\---SC are probably necessary to prove the lower bounds $(r,k)$\---SC for $r \geq 2$.

For non-constant $k$, our lower bound for $(1,k)$\---SC is $\Omega\left(C_k^{-1/5}\cdot N^{k/5}\right)$, where $C_k=\sum_{j=2}^k\frac{k!}{(j-1)!}  \leq k! \cdot e$. Notice that this term cannot be neglected for large values of $k$.
For example with $k=\log(N)$, we have $C_k \geq N$.
In comparison,  our best algorithm for $(1,k)$\---RSC, the factor in $k$ is $\binom{k}{(k+1)/2}^{-1/2}\leq \frac{2^{(k+1)/2}}{\left(\frac{k+1}{2}\cdot\pi\right)^{1/4}}$, 
which is very far from our bound on $C_k$.
It would also be interesting to see if we can tighten this gap.

\mypar{Bounds for $(r,k)$\---SC}
Unfortunately, expanding our result for the $(r,k)$\---SC problem is much more complicated than the case $r=1$ and actually even proving the case $r=2$ is not simple.
To prove such a result, one would need a bound for the problem of finding $j$ distinct $(1,k)$\---SC problem.
While proving such a bound is challenging, it is also unclear what the problem of finding $j$ distinct $(1,k)$\---SC is.
Indeed, an important property for our technique in the first lower bound proofs is that by making one query to the oracle, the adversary cannot find two or more $k$\---RSC.
The same property must hold for the problem of finding $j$ distinct $(1,k)$\---SC, and this definition and subsequent analysis remain open.

\mypar{Security of SPHINCS and SPHINCS+} 
The signature scheme SPHINCS relies on the HORST scheme (for ``HORS with trees") which adds a Merkle tree to the HORS scheme to compress the public key. The security of HORST also relies on the  $(r,k)$\---SC problem but the security of SPHINCS rely on different security notions of the underlying hash functions. In particular, it depends on a variation of the SC problem classed the \emph{target subset cover} (TSC) problem~\cite{ACISP:ReyRey02}. The main difference comes from the fact that the message signed using HORST is an unpredictable function of the actual message and this prevents an attacker to construct a subset cover beforehand. 

Nevertheless, the authors of~\cite{EC:BHHLNP15} stated an existential unforgeability result for SPHINCS~\cite[Theorem 1]{EC:BHHLNP15}  under $q_s$-adaptive chosen message attacks. The success probability in such attacks is roughly upper-bounded by:
$$
\sum_{r=1}^{\infty}\min \left(2^{r(\log{q_s}-h)+h}, 1\right) \cdot Succ_A((r,k)-SC),
$$
where $h$ is the height of the tree used in SPHINCS, and $Succ_A((r,k)-SC)$ denotes the success probability of an adversary $A$ to find a $(r,k)$\---SC. The authors made the assumption that this term is negligible for any probabilistic  adversary $A$ and our quantum lower bound on the query number to find a $(1,k)$\---SC can be seen as a first step towards proving this assumption (for idealized hash functions). To assess the security of SPHINCS from~\cite[Theorem 1]{EC:BHHLNP15} for concrete parameters such as those proposed in~\cite{EC:BHHLNP15} (namely $h=60, q_s=2^{30}$), it would also be necessary to  upper-bound the success probabilities $Succ_A((2,k)-SC)$ and $Succ_A((3,k)-SC)$, which we leave for future work.
For example, one could try to apply \cite[Theorem 4.12]{EC:YamZha21} to get a lower bound for $(r,k)$\---SC more easily, but the obtained bound will most likely not be tight.

SPHINCS+ is an enhancement of SPHINCS, which makes the scheme more efficient and
its security relies on another variant of the SC problem, namely the interleaved target subset cover (ITSC) problem. It would also be interesting to see if our methods can be used to prove similar bounds for the TSC and ITSC problems.
At last, one could also try to design algorithms for these two problems, as no quantum algorithms for them  exist yet to the best of our knowledge.

\section{Preliminaries}

We assume the reader is familiar with the theory of quantum information and for completeness, we recall Grover's algorithm and Quantum Fourier Transform (QFT) in \Cref{sec:grover_qft}.
We denote the concatenation by $||$.

\subsection{Compressed oracle technique}

We now present the key ingredients of Zhandry's compressed oracle technique, first defined in \cite{C:Zhandry19} and refined in \cite{EC:CFHL21}.
As mentioned in the introduction, the technique uses a register to keep a record of a so-called \emph{database} of the random oracle and this register is updated whenever an adversary $A$ makes a query to the random oracle.
This new register that contains the database is at the gist of our lower bounds. %

We consider the \emph{Quantum Random Oracle Model}, first defined in \cite{AC:BDFLSZ11}.
In this model, we are given black-box access to a \emph{random} function $H: \mathcal{X} \rightarrow \mathcal{Y}$.
For our model, the adversary will work on three different registers $\ket{x,y,z}$.
The first register is the query register, the second register is the answer register and the third register is the work register.
The first two registers are used for queries and answers to the oracle, while the last register is for the adversary's other computations.
We first define the unitary $\mathsf{StO}$ that represents the \emph{Standard Oracle} and that computes as follows:
$$
\mathsf{StO} \sum_{x,y,z}\alpha_{x,y,z}\ket{x,y,z} \rightarrow \sum_{x,y,z}\alpha_{x,y,z}\ket{x,y \new{+} H(x),z}
$$

This unitary corresponds to a query to $H$.

Now, we define Zhandry's compressed oracle.
In this model, instead of starting with a random function $H$, we start with the uniform superposition of all random functions $\ket{H}$, where $\ket{H}$ encodes the truth table of the function $H$.
\new{In this model, there is a register for each $x \in \mathcal{X}$, and the value of this register in the state $\ket{H}$ corresponds to $H(x)$.
That is, we have that $\ket{H} = \bigotimes_{x \in \mathcal{X}} \ket{H(x)}_x$}
Let $\mathcal{H} = \{H:\mathcal{X} \rightarrow \mathcal{Y}\}$ be the set of all possible functions $H$.
We define a new register, the database register $\ket{H}$, that starts in the uniform superposition $\frac{1}{|\mathcal{H}|}\sum_{H \in \mathcal{H}}\ket{H}$.
This register starts in product state with the other registers, and Zhandry's idea is that instead of modifying the adversary's register when querying the oracle, we will modify the database register instead.
To do so, we simply consider the \emph{Fourier basis} for the $y$ and the $H$ register before querying the Standard Oracle. 

We write this unitary $\mathsf{O}$ and it works as follows:
$$
\mathsf{O} \sum_{x,\hat{y},z}\alpha_{x,\hat{y},z}\ket{x,\hat{y},z}\otimes\sum_{\hat{H} \in \mathcal{H}}\alpha_{\hat{H}}\ket{\hat{H}} \rightarrow \sum_{x,\hat{y},z}\alpha_{x,\hat{y},z}\ket{x,\hat{y},z}\otimes\sum_{\hat{H} \in \mathcal{H}}\alpha_{\hat{H}}\ket{\hat{H} \new{\ominus (x,\hat{y})}},
$$
where, for any fixed $x\in\mathcal{X}$ and $z\in\mathcal{Y}$, $H \ominus (x,z):\mathcal{X} \rightarrow \mathcal{Y}$ is defined as:
$$
H \ominus (x,z)(x') = \begin{cases}
    H(x') &\text{ if } x' \neq x\\
    H(x) - z &\text{ if } x'=x.
\end{cases}
$$
In other words, $H\ominus(x,z)$ is obtained by replacing the value of $H(x)$ by $H(x) - z$ in $H$.

This unitary can be implemented by applying the $QFT$ to the registers $\ket{y}$ and $\ket{H}$, applying the Standard Oracle, then applying the $QFT^{\dag}$ again on the $\ket{y}$ and $\ket{H}$ registers.

Finally, we define the compression part.
The idea behind the compression is that for every $x$ in the database mapped to $\ket{\hat{0}}$, we remap it to $\ket{\bot}$, where $\bot$ is a new value outside of $\mathcal{Y}$.
More formally, the compression part is done by applying:
$$
\mathsf{Comp} = \bigotimes_x \left(\ket{\bot}\bra{\hat{0}} + \sum_{\hat{y}:\hat{y}\neq\hat{0}}\ket{\hat{y}}\bra{\hat{y}}\right)
$$
in the Fourier basis.

Since at the start of the computation, the database will be initiated with the uniform superposition over all $\mathcal{H}$ possible, then after $q$ queries the state of the database can be described with $q$ vectors.
In order to apply the compression as a unitary, we declare that $\mathsf{Comp} \ket{\bot} = \ket{0}$.

Now, we can define the \emph{Compressed Oracle}:

$$
\cO = \mathsf{Comp} \circ \mathsf{O} \circ \mathsf{Comp}^{\dag}.
$$

Of course the compression part inevitably creates some losses, compared to only using the Standard Oracle.
The precise characterization of these losses is given in one of Zhandry's lemma, and can be stated as follows:

\begin{lemma}[Lemma 5 from \cite{C:Zhandry19}]
  \label{lemma:zhandry}
  Let A be an algorithm that makes queries to a random oracle $H : \mathcal{X} \rightarrow \mathcal{Y}$, and output $(x_1,\dots,x_k,y_1,\dots,y_k) \in \mathcal{X}^k \times \mathcal{Y}^k$.
  Let $p$ be the probability that $\forall 1 \leq i \leq k$, $H(x_i) = y_i$.
  Similarly, consider the algorithm A running with the Compressed Oracle $\cO$, and output $(x'_1,\dots,x'_k,y'_1,\dots,y'_k) \in \mathcal{X}^k \times \mathcal{Y}^k$.
  Let $p'$ be the probability that $\forall 1 \leq i \leq k$, $H'(x'_i) = y'_i$\new{, where $H'$ is obtained by measuring the $H$ register at the end of the execution of the algorithm A}. Then:
  $$
  \sqrt{p} \leq \sqrt{p'} + \sqrt{\frac{k}{|\mathcal{Y}|}}
  $$
\end{lemma}

In the rest of the paper, we will have that $\sqrt{\frac{k}{|\mathcal{Y}|}}$ is negligible, and thus we will neglect this term.

\new{We also have the following lemma from \cite{EC:CFHL21} that describes the operator $\cO_{(x,\hat{y})}:\mathcal{H}\rightarrow \mathcal{H}$, which is defined as the operator applied on $\ket{H}$ when applying $\cO$ to $\ket{x}\ket{\hat{y}}\otimes\ket{H}$.
More formally, we have that:
$$
\cO \ket{x}\ket{\hat{y}}\otimes \ket{H} = \ket{x}\ket{\hat{y}}\otimes \cO_{(x,\hat{y})} \ket{H}
$$

\begin{lemma}[Lemma 4.3 from \cite{EC:CFHL21}]
    \label{lemma:matrixCO}
    For any $\hat{y} \neq \hat{0}$, the operator $\cO_{(x,\hat{y})}$ is represented by the following matrix:
    \[\setlength{\tabcolsep}{20pt} %
\renewcommand{\arraystretch}{2} %
    \begin{array}{c||c|c}
         & \bot & r\\
         \hline\hline
        \bot & 0 & \frac{\omega_N^{-ry}}{\sqrt{|\mathcal{Y}|}} \\[.2cm]
        \hline
        y' & \frac{\omega_N^{yy'}}{\sqrt{|\mathcal{Y}|}} & \begin{cases}\left(1 - \frac{2}{|\mathcal{Y}|}\right)\omega_N^{yy'} + \frac{1}{|\mathcal{Y}|} &\text{ if $y' = r$} \\
        \frac{1 - \omega_N^{yy'} - \omega_N^{ry}}{|\mathcal{Y}|} &\text{ if $y' \neq r$}
        \end{cases}
    \end{array}       
    \]
    For $\hat{y} = \hat{0}$, we have that $\cO_{(x,\hat{0})}$ is the identity.
\end{lemma}
}    
    
\new{We also define, for any compressed $H:\mathcal{X}\rightarrow \mathcal{Y\cup\{\bot\}}$, for any fixed $x\in\mathcal{X}$ and $z\in\mathcal{Y}$, $H \cup (x,z):\mathcal{X} \rightarrow \mathcal{Y}$ as:
$$
H \cup (x,z)(x') = \begin{cases}
    H(x') &\text{ if } x' \neq x\\
    z &\text{ if } x'=x.
\end{cases}
$$
In other words, $H\cup(x,z)$ is obtained by replacing the value of $H(x)$ by $z$ in $H$.
}

In the following, we will model the adversary (A) as a series of computation alternating between unitaries and oracle calls.
The adversary's quantum state will first be initialized to $\ket{0}^{\otimes N}$.
Then, his computation will be decomposed as:
\begin{equation}
  \label{equation:algorithm}
  A = U_k \cO U_{k-1} \cO \dots \cO U_2 \cO U_1 
\end{equation}

So that, if $\ket{\stateiafter} = \sum_{x,y,z,D}\alpha_{x,y,z,D}\ket{x,y,z,D}$ is the state of the adversary after $i$ quantum queries to $\cO$, then $U_{i+1}$ operates on the registers $x, y$ and $z$ only.
We also define \emph{database properties}:

\begin{definition}[Database property]
  A database property is a subset of $\mathcal{H}$.
  Any database property $D$ can be seen as a projector on $\mathcal{H}$, as follows:
  $$
  \sum_{d \in D}\ket{d}\bra{d}
  $$
\end{definition}

We write $\mathcal{D} = \{I \vert I \subseteq \mathcal{H}\}$ the set of all subspaces of $\mathcal{H}$, that also corresponds to the set of all database properties.

We now state and prove two lemmas adapted from \cite{EC:LiuZha19} that we will use thoroughly in this paper.
The first lemma will allow us to ignore the unitaries that the adversary A applies on the first registers of the state.

\begin{lemma}[adapted from Lemma 8  from \cite{EC:LiuZha19}]
  \label{lemma:commute}
  For any unitary $U$, any projector $P$, and any state $\ket{\phi}$,
  $$
  \left|\left(I \otimes P\right) \cdot \left(U \otimes I\right)\ket{\phi}\right| = \left|\left(I \otimes P\right)\ket{\phi}\right|
  $$
\end{lemma}

The second lemma bounds the amplitude of measuring a database that satisfies a property $P$ at the $i^{th}$ step of the algorithm, i.e.\ just after the $i^{th}$ query to the oracle.
In this bound, the first term captures the case where we succeed to find a database that satisfies $P$ before the $i^{th}$ query.
The second term captures the case where we did not have it before the $i^{th}$ query, but found it with the $i^{th}$ one.

\begin{lemma}[adapted from Lemma 9 from \cite{EC:LiuZha19}]
  \label{lemma:recursive}
  Let $\ket{\stateibefore}$ be the state of an algorithm A just before the $i^{th}$ quantum query to $\cO$, and $\ket{\stateiafter}$ the state of the same algorithm right after the $i^{th}$ quantum query to $\cO$.
  Let $P$ be any projector on $D$.
  We have that:
  $$
  |P\ket{\stateiafter}| \leq \left|P\ket{\stateibefore}\right| + \left|P \cO (I-P)\ket{\stateibefore}\right|
  $$
\end{lemma}

\begin{proof}
  \begin{align*}
  \left|P\ket{\stateiafter}\right| = \left|P \cO \ket{\stateibefore}\right| 
                            &= \left| P \cO (P \ket{\stateibefore} + (I - P) \ket{\stateibefore})\right| \\
                            &\leq \left| P \ket{\stateibefore} \right| +  \left| P \cO (I - P) \ket{\stateibefore})\right|,
  \end{align*}
  where the inequality comes from the triangle inequality and the fact that $P \cO P \leq P$. 
\end{proof}

\begin{rem}
    In the next section and in the rest of the paper, we will consider multiple functions $h_1,\dots,h_k : \mathcal{X} \rightarrow \mathcal{Y}$ for some fixed $k$.
    Note that this is equivalent to considering one function $H : \mathcal{X} \rightarrow \mathcal{Y}^k$, such that we interpret, for any $x \in \mathcal{X}$, the output $H(x)$ as the concatenation of values of the functions applied to $x$, i.e. $H(x) = h_1(x)||h_2(x)||\cdots||h_k(x)$.
    Hence, in this setting, the compressed oracle is used on the function $H$, and a query to any of the $h_i$ is a query to all of the $h_i$'s.
    Thus, in our results, we count the number of queries to the function $H$ and thus the number of queries to all of the $h_i$'s.
    It may seem that we lose some accuracy in this setting, however this is with the same method that multiple random functions are implemented in the literature. 
\end{rem}

\subsection{The problem of subset cover and its variants}

We define the problem of subset cover.

\begin{definition}[$(r,k)$\---SC]
  \label{def:rksc}
  Let $k,r \in \mathbb{N^*}$.
  Let $h_1,\cdots,h_k : \mathcal{X} \rightarrow \mathcal{Y}$.
  A \emph{$(r,k)$\---SC} for $(h_1,\cdots,h_k)$ is a set of $r+1$ elements $x_0,x_1,x_2,\cdots,x_r$ in $\mathcal{X}$ such that:
  $$
  \left\{h_i(x_0) \vert 1 \leq i \leq k \right\} \subseteq \bigcup_{j=1}^r \left\{h_i(x_j) \vert 1 \leq i \leq k \right\}
  $$
\end{definition}

In other words, for each $1 \leq i \leq k$, there exists a $1 \leq j \leq r$ and a $1 \leq \ell \leq k$ such that $h_i(x_0) = h_\ell(x_j)$.

We notice two facts regarding the parameters of $(r,k)$\---SC. First, we have that the problem becomes easier when $r$ increases.
Secondly, we have that when $r>k$, a $(r,k$)\---SC contains a $(k,k)$\---SC.
Thus finding a $(r,k)$\---SC when $r>k$ is the same as when $r=k$.
For simplicity, we use $k$\---SC as a shorthand of $(k,k)$\---SC.

We also define the database properties $P_{(r,k)}^{SC}$ of containing a $(r,k)$\---SC, that is the set of databases that contains a $(r,k)$\---SC.
More formally, we have that:
$$
  P_{(r,k)}^{SC} = \left\{D \in \mathcal{D} \left\vert \exists x_{0}, x_{1}, \dots, x_{r}, \forall i\neq 0, x_{0} \neq x_{i}, H(x_{0}) \subseteq \bigcup_{i=1}^r H(x_{i})  \right. \right\},
$$
where for $x \in \mathcal{X}, H(x) = \left\{h_1(x),\dots,h_k(x)\right\}$.

We follow now with the definition of a harder variation of the $k$\---subset cover called the \emph{$k$\---restricted subset cover} ($k$\---RSC).%

\begin{definition}[$k$\---RSC]
  \label{def:krsc}
  Let $k \in \mathbb{N^*}$.
  Let $h_1,\dots,h_k: \mathcal{X} \rightarrow \mathcal{Y}$.
  A \emph{$k$\---restricted subset cover} ($k$\---RSC) for $(h_1,\dots,h_k)$ is a set of $k+1$ elements $x_0,x_1,x_2,\dots,x_k$ in $\mathcal{X}$ such that:
  \begin{align*}
    \forall i \in \{1,\dots,k\}, h_i(x_0) &= h_i(x_i) 
\text{ and } x_0 \neq x_i.
  \end{align*}
\end{definition}

We also define the database properties $P_{k,\ell}^{RSC}$ of  $k$ distinct $\ell$\---RSC, that is the set of databases that contains  $k$ distinct $\ell$\---RSC.
More formally, we have that:
\begin{align}
\label{eq:Plk}
  P_{k,\ell}^{RSC} = \left\{D \in \mathcal{D} \middle\vert \begin{array}{l}
                  \exists x_{0,1}, %
                  \dots, x_{\ell,1}, \forall i\neq 0, x_{0,1} \neq x_{i,1}, \forall i,h_i(x_{0,1}) = h_i(x_{i,1})  \\
                  \exists x_{0,2}, %
                  \dots, x_{\ell,2}, \forall i\neq 0, x_{0,2} \neq x_{i,2}, \forall i,h_i(x_{0,2}) = h_i(x_{i,2})  \\
                  \vdots \\
                  \exists x_{0, \ell}, %
                  \dots, x_{\ell, k}, \forall i\neq 0, x_{0, k} \neq x_{i, k}, \forall i,h_i(x_{0, k}) = h_i(x_{i, k})  \\                    
                  \forall i \neq j, (h_1(x_{0,i}),\dots,h_{\ell}(x_{0,i})) \neq (h_1(x_{0,j}),\dots,h_{\ell}(x_{0,j}))
                  \end{array}\right\}
\end{align}

The problem of finding a $k$\---RSC was introduced in \cite{subres}, in which the authors describe an algorithm that finds a $k$\---RSC in $O\left(kN^{\frac{1}{2}\left(1-\frac{1}{2^{k+1}-1}\right)}\right)$ quantum queries to $h_1,\dots,h_k$ when the $h_i$'s are such that $|\mathcal{X}| \geq (k+1) |\mathcal{Y}|$.

We discuss now the last condition in \Cref{eq:Plk}. We remark that while such condition was not explicitly imposed in \cite{EC:LiuZha19} for their lower bound for finding multi-collisions, this property is implicitly and extensively used in their proof.
 Such a property is needed because when they count $k$\---collisions (that is, $k$ distinct $x_1,\dots,x_k$ such that $H(x_1) = \dots = H(x_k)$), they are actually interested in the number of possible \emph{images} that would be helpful to reach a $(k+1)$\---collision. In particular, this is helpful since one query can only transform \emph{one} $k$\---collision (with such a property) into a $(k+1)$\---collision.%

In our case, the last line of \eqref{eq:Plk} ensures that the ``supporting set'' of the $k$\---RSC (i.e.\ the set of images of the $x_{0,i}$ by the different random functions $h_1,\dots,h_k$) is unique.
As in the multi-collision case, this condition will be crucial to extend a $k$\---RSC to a $(k+1)$\---RSC, and for this reason we define it explicitly in $P^{RSC}_{k,\ell}$.

Finally, we state a result from \cite{EC:LiuZha19}, regarding the amplitude of finding $j$ distinct $2$\---collisions:

\begin{lemma}[adapted from \cite{EC:LiuZha19}, Corollary 11]
  \label{lemma:j2col}
  Given a random function $h: \mathcal{X} \rightarrow \mathcal{Y}$ where $|N| = \mathcal{Y}$, let $\acolh{i}{j}$ be the amplitude of the $D$ containing at least $j$ distinct $2$\---collisions after $i$ quantum queries.
  Then:
  $$
  \acolh{i}{j} \leq \left(\frac{\new{4}e\cdot i^{3/2}}{j\sqrt{N}}\right)^j.
  $$
\end{lemma}

\new{For completeness, the proof of Lemma \ref{lemma:j2col} is given in Appendix \ref{appendix:j2col}. The proof closely follows the proof of  Corollary 11 in \cite{EC:LiuZha19} but we need to consider some negligible factors that incur an extra constant factor in the statement.
}

\section{Lower bound on the $k$\---restricted subset cover problem}
\label{section:rsc}

In this section, we prove a lower bound for the $k$\---RSC problem defined in Definition \ref{def:krsc}.
This section follows closely \cite{EC:LiuZha19}'s proof of their lower bound on finding multi-collisions.
We will first prove a lower bound for the problem when $k=2$.
Then, we will prove a lower bound for finding $k$ distinct $2$\---RSC, which will be necessary in our induction step.
Finally, we will prove the induction step in the last subsection and obtain a lower bound on finding $s$ distinct $k$\---RSC.

\subsection{Finding a $2$\---restricted subset cover}
\label{subsection:2rsc}
In this section, we will prove that the number of queries necessary to find a $2$\---RSC is $\Omega(N^{3/7})$, matching the query complexity of the quantum algorithm proposed in \cite{subres}, up to a constant factor.

As presented in Definition \ref{def:krsc}, in the $2$\---RSC problem, we are given $2$ random functions $h_1, h_2$ such that for $i \in \{1,2\}$, $h_i : \mathcal{X} \rightarrow \mathcal{Y}$.
The main theorem of this subsection can be stated as follows:

\begin{theorem}
  \label{theorem:2rsc}
  Given two random functions $h_1,h_2 : \mathcal{X} \rightarrow \mathcal{Y}$ where $|N| = \mathcal{Y}$, a quantum algorithm needs to make $\Omega(N^{3/7})$ queries to $h_1$ and $h_2$ to find a $2$\---RSC with a constant probability.
\end{theorem}

In order to prove this theorem, we first introduce some database properties:
\begin{itemize}

\item $P'_{\ell-col-h_1}$ corresponds to the set of databases that contain \emph{at least} $\ell$ \emph{distinct} collisions on $h_1$.\footnote{We do not define the equivalent property for $h_2$. Since both $h_1$ and $h_2$ are random functions, we can swap them when considering database property by symmetry, thus we do not need to define more unnecessary properties.}
As explained in the previous section, here we will use the fact that we cannot reach a database containing $\ell+2$ or more collisions from a database containing $\ell$ collisions by making a single query:
$$
P'_{\ell-col-h_1} = \left\{D \in \mathcal{D} \middle\vert \begin{aligned} &\exists x_1, \dots, x_\ell, y_1, \dots, y_\ell, \forall i, h_1(x_i) = h_1(y_i) \neq \bot\\
    &\forall i, x_i \neq y_i \\ 
      &\forall i \neq j, h_1(x_i) \neq h_1(x_j)
      \end{aligned} \right\}
$$
\item $P_{\ell-col-h_1}$ corresponds to the set of databases that contain \emph{exactly} $\ell$ \emph{distinct} collisions on $h_1$:
$$
P_{\ell-col-h_1} = P'_{\ell-col-h_1} \cap \neg P'_{(\ell+1)-col-h_1} 
$$
\item $P_{preimage-h_1}$ corresponds to the set of databases that contain a preimage of $0$:\footnote{Note that the amplitude of finding \emph{any} preimage is the same as the amplitude of finding the preimage of $0$.}
$$
P_{preimage-h_1} = \left\{D \in \mathcal{D} \vert \exists x,  h_1(x) = 0 \right\}.
$$
\end{itemize}
Finally, for $i,\ell \in \mathbb{N}$, we write:
\begin{equation}
  \label{equation:fil}
  \colh{i}{\ell} = \left|P_{\ell-col-h_1}\ket{\stateiafter}\right|,\acolh{i}{\ell} = \left|P'_{\ell-col-h_1}\ket{\stateiafter}\right|, \tworsc{i} = \left|P_{1,2}^{RSC} \ket{\stateiafter}\right|,
\end{equation}
where $\ket{\stateiafter}$ is the state just after the $i^{th}$ query to $H = (h_1,h_2)$ and $P_{1,2}^{RSC}$ was defined in \Cref{eq:Plk}.
For convenience, we write $P_2 = P_{1,2}^{RSC}$ in this section.

The goal here is to bound the term $\tworsc{i}$, and to achieve this we first prove a recursive formula that involves $\colh{i}{\ell}$ as well:

\begin{lemma}
  \label{lemma:rec2rsc2}
  For every $i \in \mathbb{N}$, we have that:
  \begin{equation}
    \label{equation:rec2rsc2}
    \tworsc{i} \leq \tworsc{i-1} + \sqrt{2\sum_{\ell\geq0} \frac{\ell}{N} \colh{i-1}{\ell}^2} + \new{4}\frac{i-1}{N}.
  \end{equation}
\end{lemma}

\begin{proof}
Let $i \in \mathbb{N}$. Let $\ket{\stateibefore}$ be the state just before the $i^{th}$ query to $H = (h_1,h_2)$, namely
\begin{align*}
  \ket{\stateibefore} = \sum_{x,\hat{y},z,D} \alpha_{x,\hat{y},z,D} \ket{x,\hat{y},z} \otimes \ket{D},
\end{align*}
where $x$ is the query register, $y$ is the answer register, $z$ is the work register and $D$ is the database register.
Let $\ket{\stateiafter}$ be the state right after the  $i^{th}$ query to $H$, namely
\begin{align*}
  \ket{\stateiafter} = \sum_{\substack{x,\hat{y},z,D \\ \new{D(x) = \bot}}} \frac{1}{\sqrt{N^2}}\sum_{y'} \omega_N^{yy'}\alpha_{x,\hat{y},z,D} \ket{x,\hat{y},z} \otimes \ket{D\cup(x,y')} \new{+ \cO \sum_{\substack{x,\hat{y},z,D \\ D(x) \neq \bot}} \alpha_{x,\hat{y},z,D} \ket{x,\hat{y},z} \otimes \ket{D}}.
\end{align*}

From Lemma \ref{lemma:recursive}, we have that:

\begin{equation}
  \label{equation:rec2rsc}
  \left|P_2\ket{\stateiafter}\right| \leq \left|P_2\ket{\stateibefore}\right| + \left|P_2 \cO (I - P_2)\ket{\stateibefore}\right|.
\end{equation}

We focus now on bounding the second term:
\begin{align}
  \nonumber\left|P_2 \cO (I - P_2) \ket{\stateibefore}\right| &= \left|P_2 \cO \sum_{\substack{x,\hat{y},z \\ D : \text{ no $2$\---RSC}}} \alpha_{x,\hat{y},z,D} \ket{x,\hat{y},z,D}\right| \\ 
  \nonumber&\leq \left|P_2 \sum_{\substack{x,\hat{y},z \\ D : \text{ no $2$\---RSC}\\\new{D(x) = \bot}}} \frac{1}{\sqrt{N^2}} \sum_{y'} \omega_N^{yy'} \alpha_{x,\hat{y},z,D} \ket{x,\hat{y},z,D \cup (x, y')}\right| \\
  \nonumber&\quad\quad\new{+ \left| P_2 \cO \sum_{\substack{x,\hat{y},z \\ D : \text{ no $2$\---RSC}\\\new{D(x) \neq \bot}}} \alpha_{x,\hat{y},z,D} \ket{x,\hat{y},z} \otimes \ket{D}\right|}
\end{align}

\new{The second term can be bounded by 

  \begin{align}
    \nonumber&\left| P_2 \cO \sum_{\substack{x,\hat{y},z \\ D : \text{ no $2$\---RSC}\\\new{D(x) \neq \bot}}} \alpha_{x,\hat{y},z,D} \ket{x,\hat{y},z} \otimes \ket{D}\right| \\
    \nonumber&= \left|P_2 \sum_{\substack{x,\hat{y},z \\ D : \text{ no $2$\---RSC}\\\new{D(x) \neq \bot}}} \frac{1}{N^2}\sum_{y'} \left(1 - \omega_N^{yy'} - \omega_N^{D(x)y}\right) \alpha_{x,\hat{y},z,D} \ket{x,\hat{y},z} \otimes \ket{D \cup (x,y')}\right| \\
    \nonumber&\leq 3 \left| \frac{1}{N^2}\sum_{y'}P_2 \sum_{\substack{x,\hat{y},z \\ D : \text{ no $2$\---RSC}\\\new{D(x) \neq \bot}}} \alpha_{x,\hat{y},z,D} \ket{x,\hat{y},z} \otimes \ket{D \cup (x,y')}\right| \\
    \label{eq:boundnotbot}&\leq \frac{3 (i-1)}{N},
    \end{align}
    where the first inequality comes from Lemma \ref{lemma:matrixCO} and the fact that if the new value in the $x$ register is $\ket{\bot}$ or stays the same, then there is still no 2-RSC in $D$. The second inequality comes from the triangular inequality and the last inequality comes from using the triangular inequality and the fact that there is at most $(i-1)$ values in $D$ such that $D(x) \neq \bot$.}
  
\new{For bounding the first term,} we analyse now the possibilities for achieving a $2$\---RSC, considering the different cases of the inner sum.
We have four possible ways to get from $D$ that does not have a $2$\---RSC to $D_{y'} := D \cup (x,y')$ that has a $2$\---RSC.
\begin{itemize}
\item ($x=x_2$) Here, we consider the case where there exists an $x_0$ and $x_1$ such that $h_1(x_0) = h_1(x_1)$ and we query $x$ such that $h_2(x_0) = h_2(x)$. If we have found $\ell$ collisions of $h_1$ in $D$, then $\ell$ values of $y'$ can make $D_{y'}$ contain a $2$\---RSC, out of the $N$ possible values for the outcome of $h_2$ (notice that the value of $h_1(x)$ is not relevant for this case).
\item ($x=x_1$) Similar to the previous case, but swapping the roles of $h_1$ and $h_2$.
\item ($x=x_0$) 
Otherwise, we consider the case where we query $x$ such that we have $x_1$ and $x_2$ (which might be equal), such that $h_1(x) = h_1(x_1)$ and $h_2(x) = h_2(x_2)$. Only $i-1$ values of $y'$ will make $D_{y'}$ contain a collision on $h_1$.
  Similarly, only $i-1$ values of $y'$ will make $D_{y'}$ contain a collision on $h_2$.
\end{itemize}

Thus, we have
\begin{align}\label{eq:projection-2rsc-app}
  \left|P_2 \cO (I - P_2) \ket{\stateibefore}\right|  \leq \left(2 \cdot \sum_{\ell\geq0} \frac{\ell}{N} \left|P_{\ell-col-h_1}\ket{\stateibefore}\right|^2\right)^{1/2}
    + \new{4} \frac{(i-1)}{N},
\end{align}
and we give the details on \Cref{eq:projection-2rsc-app} in \Cref{sec:projection-2rsc-app}.

Let $\ket{\stateiminusoneafter}$ be the state just after the $(i-1)^{th}$ query, and let $U_i$ be the unitary such that $\ket{\stateibefore} = (U_i \otimes I) \ket{\stateiminusoneafter}$ (see \Cref{equation:algorithm}).
Note that we also have $\ket{\stateiafter} = \cO \cdot (U_i \otimes I) \ket{\stateiminusoneafter}$. 
Using Lemma \ref{lemma:commute}, we get that:

\begin{align}
  \nonumber\left|P_2 \cO (I - P_2) \ket{\stateibefore}\right| &\leq \sqrt{ 2 \sum_{\ell\geq0} \frac{\ell}{N} \left|P_{\ell-col-h_1}(U_i \otimes I) \ket{\stateiminusoneafter}\right|^2 } + \new{4}\frac{i-1}{N} \\
                                             &\leq \sqrt{ 2 \sum_{\ell\geq0} \frac{\ell}{N} \left|P_{\ell-col-h_1} \ket{\stateiminusoneafter}\right|^2 } + \new{4}\frac{i-1}{N}. \label{eq1}
\end{align}

Similarly, using Lemma \ref{lemma:commute}:

\begin{align}
  \label{eq2}\left|P_2\ket{\stateibefore}\right| = \left|P_2\left(U_i\otimes I\right)\ket{\stateiminusoneafter}\right| = \left|P_2\ket{\stateiminusoneafter}\right|.
\end{align}

Then, using \Cref{equation:rec2rsc}, \Cref{eq1} and \Cref{eq2}, and the notation from \Cref{equation:fil}, we have:

\begin{align*}
  \tworsc{i} \leq \tworsc{i-1} + \sqrt{2\sum_{\ell\geq0} \frac{\ell}{N} \colh{i-1}{\ell}^2} + \new{4}\frac{i-1}{N}.
\end{align*}
\end{proof}

We will now expand this recursive formula to obtain a bound on $\tworsc{i}$.
\begin{lemma}
  \label{lemma:bound2rsc}
  For every $i \in \mathbb{N}$, we have that:
  \begin{align*}
    \tworsc{i} \leq \sqrt{2}\sum_{j=1}^{i-1}\sqrt{\frac{\mu_3(j)}{N}} + \sqrt{2}\cdot 2^{-9.5N^{1/8}} + \new{4}\frac{i^2}{N},
  \end{align*}
  where
  $$\mu_3(j) = \max\left\{{\new{8}e\frac{j^{3/2}}{\sqrt{N}}, 10N^{1/8}}\right\}.$$
\end{lemma}

\begin{proof}
  From Lemma \ref{lemma:rec2rsc2}, we expand recursively \Cref{equation:rec2rsc2}, and obtain (using that $\tworsc{0} = 0$):

  \begin{equation}
  \label{eq:gi}
    \tworsc{i} \leq \sum_{j=1}^{i-1} \sqrt{2\sum_{\ell\geq0} \frac{\ell}{N} \colh{j}{\ell}^2} + \new{4}\sum_{j=1}^{i-1}\frac{j}{N}.
  \end{equation}
  The second term of \Cref{eq:gi} can be bounded by
  \begin{equation}
  \label{eq:i2n}
     \new{4}\sum_{j=1}^{i-1}\frac{j}{N} \leq \new{4}\sum_{j=1}^{i-1}\frac{i}{N}
    \leq  \new{4}\frac{i^2}{N}.
  \end{equation}
  As for the first term of \Cref{eq:gi}, we have:

  \begin{align}
    \nonumber\sum_{j=1}^{i-1} \sqrt{2\sum_{\ell\geq0} \frac{\ell}{N} \colh{j}{\ell}^2} &= \sqrt{2} \sum_{j=1}^{i-1} \sqrt{\sum_{\ell=0}^{\mu_3(j)} \frac{\ell}{N} \colh{j}{\ell}^2 + \sum_{\ell > \mu_3(j)} \frac{\ell}{N} \colh{j}{\ell}^2} \\
                                                                 \nonumber &\leq \sqrt{2} \sum_{j=1}^{i-1} \left(\sqrt{\sum_{\ell=0}^{\mu_3(j)} \frac{\ell}{N} \colh{j}{\ell}^2}
                                                                   + \sqrt{\sum_{\ell > \mu_3(j)} 1\cdot \colh{j}{\ell}^2}\right) \\
                                                                 \nonumber &\leq \sqrt{2} \sum_{j=1}^{i-1} \left(\sqrt{\frac{\mu_3(j)}{N}\cdot \acolh{j}{1}}
                                                                   + \acolh{j}{\mu_3(j)}\right) \\
                                                                 \label{eq:fcol}&\leq \sqrt{2} \left(\sum_{j=1}^{i-1} \sqrt{\frac{\mu_3(j)}{N}}
                                                                   + \sum_{j=1}^{i-1} \acolh{j}{\mu_3(j)}\right)
  \end{align}
  where in the second inequality, we used the fact that the term $\sum_{\ell>\mu_3(j)}\colh{j}{\ell}^2$ is equal to the amplitude of finding \emph{at least} $\mu_3(j)$ distinct $\ell$\---collisions on $h_1$, thus is exactly equal to $\acolh{j}{\mu_3(j)}^{2}$ (defined in \Cref{equation:fil}), and similarly for $\acolh{j}{1}$.
    
  It follows that
  \begin{equation}
  \label{eq:expnegl}
    \sum_{j=1}^{i-1}\acolh{j}{\mu_3(j)} \leq \sum_{j=1}^{i-1}\left(\frac{\new{4}e\cdot j^{3/2}}{\mu_3(j)\cdot\sqrt{N}}\right)^{\mu_3(j)} %
                                    \leq \sum_{j=1}^{i-1} \left(\frac{1}{2}\right)^{10N^{1/8}}
                                    \leq 2^{-9.5N^{1/8}},
  \end{equation}
  where the first inequality comes from Lemma \ref{lemma:j2col}, the second inequality comes from the definition of $\mu_3(j)$ and in the last inequality we assume that $i\leq N^{1/2}$. Indeed, otherwise $A$ can execute \cite{subres}'s algorithm whose query complexity for finding a $k$\---RSC is upper-bounded by $O\left(N^{1/2}\right)$.

Putting together \Cref{eq:gi}, \Cref{eq:i2n}, \Cref{eq:fcol} and \Cref{eq:expnegl} gives the result.
\end{proof}

We can now use Lemma \ref{lemma:bound2rsc} to prove Theorem \ref{theorem:2rsc}

\begin{proof}[Proof of Theorem \ref{theorem:2rsc}]
  Using Lemma \ref{lemma:bound2rsc}, we have for $i \in \mathbb{N}$:

$$\tworsc{i} \leq \sqrt{2}\sum_{j=1}^{i-1}\sqrt{\frac{\mu_3(j)}{N}} + \sqrt{2}\cdot 2^{-9.5N^{1/8}} + \new{4} \frac{i^2}{N}.$$  

We can bound the first term by:
\begin{align*}
  \sqrt{2}\sum_{j=1}^{i-1}\sqrt{\frac{\mu_3(j)}{N}} &= \sqrt{2}\left(\sum_{j : \mu_3(j) = \new{8}e\cdot\frac{j^{3/2}}{\sqrt{N}}}\frac{\sqrt{\new{8}ej^{3/2}}}{N^{3/4}} + \sum_{j : \mu_3(j) = 10N^{1/8}} \frac{\sqrt{10N^{1/8}}}{N^{1/2}}\right) \\
                                                    &\leq \sqrt{2}\left(\sum_{j=1}^{i-1}\frac{\sqrt{\new{8}ej^{3/2}}}{N^{3/4}} + \sum_{j : \mu_3(j) = 10N^{1/8}} \frac{\sqrt{10N^{1/8}}}{N^{1/2}}\right) \\
                                                    &\leq \new{4}\sqrt{e}\frac{i^{7/4}}{N^{3/4}} + \left(\frac{10}{\new{8}e}\right)^{2/3} \cdot N^{5/12} \cdot \frac{\sqrt{10N^{1/8}}}{N^{1/2}} \\
                                                    &\leq \new{4}\sqrt{e}\frac{i^{7/4}}{N^{3/4}} + O(N^{-1/48}),
\end{align*}
where the second inequality comes from counting the number of $j$ such that $\mu_3(j)=10N^{1/8}$, which is equal to the number of $j$ such that $\new{8}e\frac{j^{3/2}}{\sqrt{N}} \leq 10N^{1/8}$.

Thus, we have the following bound on $\tworsc{i}$:

$$\tworsc{i} \leq \new{4} \sqrt{e}\frac{i^{7/4}}{N^{3/4}} + \new{4}\frac{i^2}{N} + O(N^{-1/48}).$$

This bound is in the compressed oracle model, and using Lemma \ref{lemma:zhandry} we obtain the same bound in the random oracle model by putting the negligible term in the $O\left(N^{-1/48}\right)$.

So when $i = o(N^{3/7})$, we have $\tworsc{i} = o(1)$. Hence if we want $\tworsc{i}$ to be constant, i.e.\ not $o(1)$, we must have $i = \Omega\left(N^{3/7}\right)$.
\end{proof}

\subsection{Finding $k$ distinct $2$-restricted subset cover}
\label{subsection:k2rsc}

We are now interested in bounding the number of queries needed to find $k$ distinct triplets that satisfy a $2$\---RSC.
We have the following result:
\begin{theorem}
  \label{theorem:k2rsc}
  Given two random functions $h_1,h_2 : \mathcal{X} \rightarrow \mathcal{Y}$ where $N = |\mathcal{Y}|$, a quantum algorithm needs to make $\Omega(k^{4/7} \cdot N^{3/7})$ queries to $h_1$ and $h_2$ to find k distinct $2$\---RSC with constant probability, for any $k \leq N^{1/8}$.
\end{theorem}

To prove this theorem, we first introduce some notation.
We denote $P_{2,k,\ell}$ the set of databases that satisfies $k$ distinct $2$\---RSC, and that contain exactly $\ell$ collisions on $h_1$.
Using the notation from the Section \ref{subsection:2rsc} and \Cref{eq:Plk}, we have that $P_{2,k,\ell} = P_{k,2}^{RSC} \cap P_{\ell-col-h_1}$.
We denote $\ktworsc{i}{k} = \left|P_{k,2}^{RSC}\ket{\stateiafter}\right|$ and $\ktworsclcol{i}{k}{\ell} = \left|P_{2,k,\ell}\ket{\stateiafter}\right|$, where $\ket{\stateiafter}$ is the state just after the $i^{th}$ query to $H = (h_1,h_2)$.

Our goal is to bound $\ktworsc{i}{k}$, and as in the previous subsection, we will first prove a recursive formula stated in the next lemma.
\begin{lemma}
  \label{lemma:reck2rsc}
  For every $i \in \mathbb{N}$, and every $k \in \mathbb{N}$, we have that:
$$ \ktworsc{i}{k} \leq \ktworsc{i-1}{k} + \sqrt{2 \sum_{\ell\geq0} \frac{\ell}{N} \ktworsclcol{i-1}{k-1}{\ell}^2} +  \frac{(i-1)}{N} \ktworsc{i-1}{k-1}.$$
\end{lemma}

\begin{proof}
From Lemma \ref{lemma:recursive}, we have the following inequality:

$$\left|P_{k,2}^{RSC}\ket{\stateiafter}\right| \leq \left|P_{k,2}^{RSC}\ket{\stateibefore}\right| + \left|P_{k,2}^{RSC} \cO (I - P_{k,2}^{RSC}) \ket{\stateibefore}\right|.$$

And we have that:
\begin{align*}
  &\left|P_{k,2}^{RSC} \cO (I-P_{k,2}^{RSC})\ket{\stateibefore}\right| \\
  &\leq  \left|P_{k,2}^{RSC} \sum_{\substack{x,\hat{y},z \\ D : \text{k-1 2-RSC}\\ \new{D(x) = \bot}}} \frac{1}{\sqrt{N^2}} \sum_{y'} \omega_N^{yy'} \alpha_{x,\hat{y},z,D} \ket{x,\hat{y},z,D \cup (x, y')}\right|\new{+ 3 \frac{i-1}{N} \left|P_{k-1,2}^{RSC}\ket{\stateibefore}\right|^2}\\
  &\leq \left(2 \sum_{\ell\geq0} \frac{\ell}{N}\sum_{\substack{x,\hat{y},z \\ D : \text{k-1 2-RSC} \\ \text{$\ell$ collisions}\\\text{on $h_1$}}} \left|\alpha_{x,\hat{y},z,D}\right|^2 \right)^{1/2}
  + \left(\frac{(i-1)^2}{N^2}\sum_{\substack{x,\hat{y},z \\ D : \text{k-1 2-RSC}}} \left|\alpha_{x,\hat{y},z,D}\right|^2\right)^{1/2} \new{+ 3 \frac{i-1}{N} \left|P_{k-1,2}^{RSC}\ket{\stateibefore}\right|^2}\\
  &\leq \left(2 \sum_{\ell\geq0} \frac{\ell}{N} \left|P_{2,k-1,\ell}\ket{\stateibefore}\right|^2\right)^{1/2}
    \new{+ 4 \frac{i-1}{N} \left|P_{k-1,2}^{RSC}\ket{\stateibefore}\right|^2},
\end{align*}
\new{where the first inequality comes from the same calculations done to obtain \Cref{eq:boundnotbot}, and} the second equality uses the same cases as for the case $k=1$ in Lemma~\ref{lemma:rec2rsc2}.%

Using Lemma \ref{lemma:commute} and previous notation (as in Lemma \ref{lemma:rec2rsc2}), we obtain that:

\begin{align*}
  \ktworsc{i}{k}
          &\leq \ktworsc{i-1}{k}
            + \left(2 \sum_{\ell\geq0} \frac{\ell}{N} \ktworsclcol{i-1}{k-1}{\ell}^2\right)^{1/2}
            + \new{4} \frac{(i-1)}{N} \ktworsc{i-1}{k-1}.
\end{align*}
\end{proof}

Following the proof from  the case $k=1$, we will split the sum in two using $\mu_3(j)$ as a threshold.
We also define a new notation that will simplify expressions:
\begin{definition}
  $$A_i = \sum_{\ell=0}^{i-1}\sqrt{2}\left(\sqrt{\frac{\mu_3(\ell-1)}{N}} + \new{\sqrt{8}}\frac{\ell-1}{N}\right),$$
  where
  $$\mu_3(\ell) = \max\left\{{\new{8}e\frac{\ell^{3/2}}{\sqrt{N}}, 10N^{1/8}}\right\}.$$
\end{definition}

Before bounding $\ktworsc{i}{k}$, we first prove a bound on $A_i$.
\begin{lemma}
  \label{lemma:boundAi}
  For every $i \in \mathbb{N}$, we have that:
  $$ A_i \leq \new{8} \sqrt{e} \frac{i^{7/4}}{N^{3/4}} + \new{4} \frac{i^2}{N} + O\left(N^{-1/48}\right).$$
  It follows that $A_i < 2eN^{1/8}$ for $i \leq N^{1/2}$.
\end{lemma}

We leave the proof of Lemma~\ref{lemma:boundAi} to \Cref{sec:boundAi}.
We can now state the lemma that bounds $\ktworsc{i}{k}$.

\begin{lemma}
  \label{lemma:born2rsc}
  For every $i \in \mathbb{N}$ and $k \in \mathbb{N}$, we have that:
  $$
  \ktworsc{i}{k} < \frac{A_i^k}{k!} + \sqrt{2}\cdot 2^{-N^{1/8}}.
  $$
\end{lemma}

\begin{proof}

  We write $\acolh{i}{j} = \left|P'_{j-col-h_1}\ket{\stateibefore}\right|$. From Lemma \ref{lemma:reck2rsc}, we have that:
\begin{align}
  \nonumber\ktworsc{i}{k} &\leq \ktworsc{i-1}{k} + \sqrt{2 \sum_{\ell\geq0} \frac{\ell}{N} \cdot \ktworsclcol{i-1}{k-1}{\ell}^2} + \new{4}\frac{i-1}{N} \cdot \ktworsc{i-1}{k-1} \\
          \nonumber&\leq \ktworsc{i-1}{k} + \sqrt{2}\left(\sqrt{\frac{\mu_3(i-1)}{N}} \cdot \ktworsc{i-1}{k-1} + \acolh{i-1}{\mu_3(i-1)}\right) + \new{4} \frac{i-1}{N} \cdot \ktworsc{i-1}{k-1} \\
          \label{equation:beforebound2sc}&= \ktworsc{i-1}{k} + \sqrt{2}\left(\sqrt{\frac{\mu_3(i-1)}{N}} + \new{\sqrt{8}}\frac{i-1}{N}\right)\ktworsc{i-1}{k-1} + \sqrt{2}\cdot \acolh{i-1}{\mu_3(i-1)},
  \end{align}
  where the second inequality comes from separating the sum in two, similar to the proof of Lemma \ref{lemma:bound2rsc}.
  
  Following \cite{EC:LiuZha19}'s proof for Lemma 14, by expanding the recursion we get:

  \begin{equation}\label{equation:bound2sc}
  \ktworsc{i}{k} \leq \frac{A_i^k}{k!} + \sqrt{2} \cdot e^{A_i}  2^{9.5N^{1/8}}.
  \end{equation}

  For completeness, the proof of \Cref{equation:bound2sc} is given in Appendix \ref{appendix:recursion}. Using Lemma \ref{lemma:boundAi}, we can bound the second term, and:
  $$\ktworsc{i}{k} < \frac{A_i^k}{k!} + \sqrt{2}\cdot 2^{-N^{1/8}}.$$
\end{proof}

We can now prove the main theorem of this subsection.

\begin{proof}[Proof of Theorem \ref{theorem:k2rsc}]
  Following from Lemma \ref{lemma:born2rsc}, we have that:
  
  \begin{align*}
    \ktworsc{i}{k} &\leq \frac{A_i^k}{k!} + \sqrt{2}\cdot 2^{-{N^{1/8}}} \leq \left(\frac{A_i \cdot e}{k}\right)^k + \sqrt{2}\cdot 2^{-N^{1/8}}.
  \end{align*}

  We now use the bound on $A_i$ of Lemma \ref{lemma:boundAi}:
  $$
  \ktworsc{i}{k} \leq \left(\frac{\new{8}e^{3/2}}{k} \cdot \frac{i^{7/4}}{N^{3/4}} + \frac{\new{4}e}{k}\cdot\frac{i^2}{N} + \frac{e}{k} \cdot O\left(N^{-1/48}\right)\right)^k + \sqrt{2}\cdot 2^{-N^{1/8}}.
  $$
  So if $i=o(k^{4/7}\cdot N^{3/7})$, then $\ktworsc{i}{k}=o(1)$.
  Hence if we want $\ktworsc{i}{k}$ to be a constant, i.e.\ not $o(1)$, we must have $i=\Omega\left(k^{4/7}\cdot N^{3/7}\right)$.
\end{proof}

\subsection{Finding $k$ distinct $s$-restricted subset cover}

In this section, we generalize the result to the problem of finding $k$ distinct $s$\---RSC, for any $s \geq 3$ and any $k \geq 1$.
We are given $s$ random functions $h_1, \dots, h_s$ such that for any $i \in [1,s]$, $h_i : \mathcal{X} \rightarrow \mathcal{Y}$.
We will prove the following theorem.

\begin{theorem}
\label{theorem:krsc}
  Given $s$ random functions $h_1,\dots,h_s: \mathcal{X} \rightarrow \mathcal{Y}$ where $N=|\mathcal{Y}|$, a quantum algorithm needs to make $\Omega\left((s+1)^{-\frac{2^{s}}{2^{s+1}-1}}\cdot k^{\frac{2^{s}}{2^{s+1}-1}}\cdot N^{\frac{2^{s}-1}{2^{s+1}-1}}\right)$ queries to $h_1,\dots,h_s$ to find $k$ distinct $s$\---RSC with constant probability, for any $s\leq \log (\log (N))$ and any $k \geq N^{1/2^{s+1}}$.
\end{theorem}

And naturally we have the following corollary for $k=1$:

\begin{corollary}
  Given $s$ random functions $h_1,\dots,h_s: \mathcal{X} \rightarrow \mathcal{Y}$ where $N=|\mathcal{Y}|$, a quantum algorithm needs to make $\Omega\left((s+1)^{-\frac{2^{s}}{2^{s+1}-1}}\cdot N^{\frac{2^{s}-1}{2^{s+1}-1}}\right)$ queries to $h_1,\dots,h_s$ to find one $s$\---RSC with constant probability, for any $s\leq \log (\log (N))$.
\end{corollary}

In order to prove Theorem \ref{theorem:krsc}, we first define some notations, starting with the notations for the amplitudes. We define:
\begin{enumerate}
\item $\sminusonersc{i}{j}$ as the amplitude of the databases $D$ containing at least $j$ distinct $(s-1)$\---RSC after $i$ quantum queries.
\item $\ssrsc{i}{j}{k}$ as the amplitude of the databases $D$ containing at least $j$ distinct $(s-1)$\---RSC and exactly $k$ distinct $s$\---RSC after $i$ quantum queries.
\item $\srsc{i}{k}$ as the amplitude of the databases $D$ containing exactly $k$ distinct $s$\---RSC after $i$ quantum queries.
\end{enumerate}

More formally, let $\ket{\stateibefore}$ (resp. $\ket{\stateiafter}$) be the state of the algorithm just before (resp. after) the $i^{th}$ query to the oracle. We have:
\begin{align*}
&\sminusonersc{i}{j} = \left|P_{j,(s-1)}^{RSC}\ket{\stateiafter}\right|, \\
&\ssrsc{i}{j}{k} = \left|P_{j,(s-1)}^{RSC}P_{k,s}^{RSC}\neg P_{k+1,s}^{RSC}\ket{\stateiafter}\right|, \\
&\srsc{i}{k} = \left|P_{k,s}^{RSC}\neg P_{k+1,s}^{RSC}\ket{\stateiafter}\right|.
\end{align*}

We want to bound $\srsc{i}{k}$, and to do so, we define some convenient notation.
We start by defining $\Pi_s$, a term that appears in the bound of $\srsc{i}{k}$.

\begin{definition}
Let $\Pi_s$ be defined as follows:

\begin{align*}
  \begin{cases}
    \Pi_1 = 1 \\
    \Pi_2 = 1 \\
    \forall s \geq 2, & \Pi_{s+1} = 2 \cdot \sqrt{s} \cdot \sqrt{\Pi_s}
  \end{cases}
\end{align*}

\end{definition}

We define $A_{i,s}$ and $\mu_s(\ell)$ as follows:
\begin{definition}
  \label{def:Ais}
  $$A_{i,s} =  \sum_{\ell=0}^{i-1}B_{\ell,s-1},$$
  where
  $$
  B_{\ell,s} = \sqrt{s \cdot \frac{\mu_{s+1}(\ell)}{N}} + \new{4}\left(\frac{\ell}{N}\right)^{s/2} + \left(\sum_{r=2}^s \frac{\ell}{N^r}\right)^{1/2},
  $$
  and
  $$\mu_s(\ell) = \max \left\{ \Pi_{s-1} \cdot (\new{8}e)^{\frac{2^{s-2}-1}{2^{s-3}}}\frac{\ell^{(2^{s-1}-1)/2^{s-2}}}{N^{(2^{s-2}-1)/2^{s-2}}} , \new{40} \cdot s^2 \cdot \Pi_{s-1} \cdot N^{1/2^s} \right\}.$$
\end{definition}

We can now state the bound on $\srsc{i}{k}$ that we will need to prove Theorem \ref{theorem:krsc}:
\begin{lemma}
  \label{lemma:krscgik}
  For every $i \in \mathbb{N}$ and every $k \in \mathbb{N}$, we have that:
  $$
  \srsc{i}{k} \leq \frac{A_{i,s+1}^k}{k!} + O\left(2^{- (s+1)^2 \cdot \Pi_s \cdot N^{1/2^{s+1}}}\right).
  $$
\end{lemma}

In order to prove Lemma~\ref{lemma:krscgik}, we first prove a bound on $A_{i,s}$.

\begin{lemma}
  \label{lemma:Aisbound}

  $A_{i,s} \leq (\new{8}e)^{\frac{2^{s-2} - 1}{2^{s-2}}} \frac{i^{(2^s-1)/2^{s-1}}}{N^{(2^{s-1}-1)/2^{s-1}}} \cdot \Pi_s + O\left( s^4 \cdot \Pi_s \cdot N^{-1/(2^s(2^s-2))}\right)$
\end{lemma}

In the interest of space, we leave the proof of Lemma~\ref{lemma:Aisbound} to \Cref{sec:Aisbound}, and we now prove Lemma~\ref{lemma:krscgik}.

\begin{proof}[Proof of Lemma \ref{lemma:krscgik}]
  We prove this theorem by induction.
  The case $s=3$ corresponds to the subsection \ref{subsection:k2rsc}.
  Fix $s\geq3$.
  We assume that $\sminusonersc{i}{j} \leq \frac{A_{i,s}^j}{j!} + O\left(2^{-s^2 \cdot \Pi_{s-1} \cdot N^{1/2^s}}\right)$ for every $i \in \mathbb{N}$ and $j \in \mathbb{N}$.
  We will show that $\srsc{i}{k} \leq \frac{A_{i,s+1}^k}{k!} + O\left(2^{-(s+1)^2 \cdot \Pi_s \cdot N^{1/2^{s+1}}}\right)$.

  Similarly to the previous subsection, we will bound $\srsc{i}{k}$ recursively.
  Using Lemma \ref{lemma:recursive}, we have that:
  \begin{align*}
    \left|P_{k,s}^{RSC}\ket{\stateiafter}\right| &\leq \left|P_{k,s}^{RSC}\ket{\stateibefore}\right| + \left|P_{k,s}^{RSC} \cO \left(I - P_{k,s}^{RSC}\right)\ket{\stateibefore}\right|,
  \end{align*}
  where the second term can be written as:
  \begin{align}
      \label{eq:ksrscbot}\left|P_{k,s}^{RSC}\sum_{\substack{x,\hat{y},z\\D:(k-1) \text{ distinct } s-RSC\\\new{D(x) = \bot}}} \frac{1}{\sqrt{N^s}} \sum_{y'}\omega_{n}^{yy'}\alpha_{x,\hat{y},z,D}\ket{x,\hat{y},z,D\cup(x,y')}\right| \\
    \label{eq:ksrscnotbot} \new{+ \left|P_{k,s}^{RSC}\cO\sum_{\substack{x,\hat{y},z\\D:(k-1) \text{ distinct } s-RSC\\\new{D(x) \neq \bot}}} \alpha_{x,\hat{y},z,D}\ket{x,\hat{y},z,D}\right|.}
  \end{align}

  To bound the term of \Cref{eq:ksrscbot}, we analyse now the possibilities for achieving a $s$\---RSC, considering the different cases of the inner sum.
  We have different possible ways to get from $D$ that does not have a $s$-RSC to $D_{y'} := D \cup (x,y')$ that has a $s$-RSC.
  \begin{itemize}
  \item ($x=x_0$) As for the case $s=2$, we consider the cases where we query $x$ such that we have $x_1,\dots,x_s$, such that $\forall 1\leq j\leq s$, $h_s(x)=h_s(x_s)$.
    For every $1\leq j\leq s$, only $i-1$ values of $y'$ will make $D_{y'}$ contain a collision on $h_s$.
    Thus there are at most $\frac{(i-1)^s}{N^s}$ values of $y'$ such that $D_{y'}$ contain a new $s$\---RSC in this case.
  \item ($x=x_{s}$) Similarly to the case $s=2$, we consider the case where there exists $x_0,\dots,x_{s-1}$ such that $x_0,\dots,x_{s-1}$ is a $(s-1)$\---RSC, and we query $x$ such that $h_j(x)=h_j(x_0)$ for some $1\leq j\leq s$.
    If we have found $\ell$ distinct $(s-1)$\---RSC in $D$ previously, then $\ell$ values of $y'$ can make $D_{y'}$ contain a $s$\---RSC, out of the N possible values for the outcome of $h_j$ (notice that the values of $h_i(x)$ for $i\neq j$ are not relevant for this case), and there are $s$ different values for $j$.
  \item However, some new terms do not appear in the case of $2$\---RSC.
    That would be the case where the query $x$ is equal to $x_{i_1} = x_{i_2} = \dots = x_{i_r}$ for some $r \in \{2,\dots,s \}$ in the new $s$\---RSC.
    We bound these terms as follows: for each $r$, there is at most $(i-1)$ distinct $(s-r)$\---RSC.
    For each of these $(s-r)$\---RSC, there are $r$ collisions missing on some $h_{i_1}, \dots, h_{i_r}$.
    And exactly one value of $y'$ will make $D_{y'}$ contain a collision for $h_{i_j}$.
    The values of the other hash functions are irrelevant here.
    Hence using Lemma \ref{lemma:commute} we can bound the probability of this event by:
    \begin{align}\label{eq:case3}
        \sum_{r=2}^s \frac{i-1}{N^r} \srsc{i-1}{k-1}^2,
    \end{align}
    where we bound the amplitude of the database containing at least one $(s-r)$\---RSC and $k-1$ distinct $s$\---RSC after $i-1$ quantum queries by $\srsc{i-1}{k-1}$, the amplitude of the databases containing only $k-1$ distinct $s$\---RSC after $i-1$ quantum queries.
  \end{itemize}

Using Lemma~\ref{lemma:commute}, \new{and as for the previous cases, by bounding the term of \Cref{eq:ksrscnotbot} by $3 \left(\frac{(i-1)}{N}\right)^{s/2} \srsc{i-1}{k-1}$, }we can upper bound $\srsc{i}{k}$ by
\begin{align}
  \nonumber
&  \srsc{i-1}{k} + \sqrt{
            s \sum_{\ell\geq0} \frac{\ell}{N} \ssrsc{i-1}{\ell}{k-1}^2}
            + \new{4}\sqrt{\frac{(i-1)^s}{N^s} \srsc{i-1}{k-1}^2} +
            \sqrt{\sum_{r=2}^s \frac{i-1}{N^r}\srsc{i-1}{k-1}^2} \\
          \label{eq:gik}&\leq \srsc{i-1}{k} 
          + \sqrt{s \sum_{\ell\geq0} \frac{\ell}{N} \ssrsc{i-1}{\ell}{k-1}^2}
            + \left(\new{4}\left(\frac{i-1}{N}\right)^{s/2}
            + \left(\sum_{r=2}^s \frac{i-1}{N^r}\right)^{1/2}\right)\srsc{i-1}{k-1},
\end{align}
where the second term can be split in two, similar to the proof of Lemma \ref{lemma:bound2rsc}:

\begin{align*}
    \sqrt{s \sum_{\ell\geq0} \frac{\ell}{N} \ssrsc{i-1}{\ell}{k-1}^2} &\leq \sqrt{s \cdot \frac{\mu_{s+1}(i-1)}{N}}\srsc{i-1}{k-1} + \sqrt{s} \cdot \sminusonersc{i-1}{\mu_{s+1}(i-1)}
\end{align*}

The term $\sminusonersc{i-1}{\mu_{s+1}(i-1)}$ can be bounded by induction hypothesis by:
\begin{align*}
  \sminusonersc{i-1}{\mu_{s+1}(i-1)} &\leq \frac{A_{i-1,s}^{\mu_{s+1}(i-1)}}{\mu_{s+1}(i-1)!} + O\left(2^{-s^2 \cdot \Pi_{s-1} \cdot N^{1/2^s}}\right),
\end{align*}
  and the first term can be bounded by using Lemma \ref{lemma:Aisbound} and the definition of $\mu_{s+1}(i-1)$ by:
\begin{align*}
                      \left( \frac{e(\new{4}e)^{\frac{2^{s-2}-1}{2^{s-2}}} \frac{i^{(2^s-1)/2^{s-1}}}{N^{(2^{s-1}-1)/2^{s-1}}} \Pi_s
                         + O\left(s^4\Pi_s N^{-1/(2^s(2^s-2))}\right)}
                         {\max \left\{ (\new{8}e)^{\frac{2^{s-1}-1}{2^{s-2}}} \frac{i^{(2^{s}-1)/2^{s-1}}}{N^{(2^{s-1}-1)/2^{s-1}}} \Pi_s,
                         \new{40} (s+1)^2 \Pi_s \cdot N^{1/2^s} \right\}}\right)^{\new{40} (s+1)^2 \Pi_s N^{1/2^{s+1}}},
\end{align*}
which is smaller than
\begin{align*}
                       \left(\frac{1}{2} + o(1)\right)^{\new{40} (s+1)^2 \cdot \Pi_s \cdot N^{1/2^{s+1}}},
\end{align*}
which leads to:
$$
\sminusonersc{i-1}{\mu_{s+1}(i-1)} < 2^{-9.8 \new{\cdot 4} \cdot (s+1)^2 \cdot \Pi_s \cdot N^{1/2^{s+1}}}.
$$

Using Definition \ref{def:Ais}, we rewrite \Cref{eq:gik} as:
$$
\srsc{i}{k} \leq \srsc{i-1}{k} + B_{\ell,s} \cdot \srsc{i-1}{k-1} + \sqrt{s} \cdot 2^{-9.8 \cdot \new{4} \cdot (s+1)^2\cdot\Pi_s\cdot N^{1/2^{s+1}}}.
$$

Then, by expanding the inequality and using the fact that $\srsc{0}{k-1}=0$, we get:
\begin{align*}
  \srsc{i}{k} \leq& \srsc{i-1}{k} + B_{\ell,s} \cdot \srsc{i-1}{k-1} + \sqrt{s} \cdot 2^{-9.8 \new{\cdot 4}\cdot (s+1)^2\cdot\Pi_s\cdot N^{1/2^{s+1}}} \\
  \vdots& \\
  \leq& \sum_{\ell=0}^{i-1}\left(B_{\ell,s} \cdot \srsc{\ell}{k-1} + \sqrt{s} \cdot 2^{-9.8 \new{\cdot 4}\cdot (s+1)^2 \cdot \Pi_s \cdot N^{1/2^{s+1}}} \right)\\
  \leq& \left(\sum_{\ell=0}^{i-1}B_{\ell,s} \cdot \srsc{\ell}{k-1}\right) + s \cdot N^{1/2} \cdot \sqrt{s} \cdot 2^{-9.8 \new{\cdot 4}\cdot (s+1)^2 \cdot \Pi_s \cdot N^{1/2^{s+1}}} \\
  \leq& \left(\sum_{\ell=0}^{i-1}B_{\ell,s} \cdot \srsc{\ell}{k-1}\right) + s^{3/2} \cdot 2^{-9.5 \new{\cdot 4} \cdot (s+1)^2 \cdot \Pi_s \cdot N^{1/2^{s+1}}},
\end{align*}
where we use the fact that $i\leq s\cdot\sqrt{N}$ for the third inequality.

Expanding this inequality, we obtain
\begin{align}\label{eq:bound-giksrsc-app}
  \srsc{i}{k} \leq
   \frac{A_{i,s+1}^k}{k!} + s^{3/2} \cdot e^{A_{i,s+1}}\cdot 2^{-9.5 \new{\cdot 4}\cdot (s+1)^2 \cdot \Pi_s \cdot N^{1/2^{s+1}}}.
\end{align}

For details on \Cref{eq:bound-giksrsc-app}, see \Cref{sec:eq:bound-giksrsc-app}.

And because $i \leq s\cdot \sqrt{N}$, we have $A_{i,s+1} \leq  \new{8}e \cdot (s+1)^2 \cdot \Pi_s \cdot N^{1/2^{s+1}}$. Using this and the fact that $s^{3/2} \leq 2^{\Pi_s \cdot (s+1)^2 \cdot N^{1/2^{s+1}}}$, we conclude:

$$
\srsc{i}{k} \leq \frac{A_{i,s+1}^k}{k!} +  2^{-(s+1)^2 \cdot \Pi_s \cdot N^{1/2^{s+1}}}.
$$
\end{proof}

At last we bound $\Pi_s$ to conclude the analysis.
\begin{proposition}
  \label{proposition:Pi}
  We have for any $s \in \mathbb{N}$ that:
  $$
  \Pi_s \leq 4 s
  $$
\end{proposition}

\begin{proof}
  The statement is true for $s=1,2$.  Assume it is true for $s\geq2$. Then,
  \begin{align*}
    \Pi_{s+1} &= 2 \sqrt{s} \cdot \sqrt{\Pi_s} 
              \leq 2 \sqrt{s} \cdot \sqrt{4s} 
              \leq 4 (s+1).
  \end{align*}
\end{proof}

Finally, we can prove Theorem \ref{theorem:krsc}:

\begin{proof}[Proof of Theorem \ref{theorem:krsc}]
  From Lemma \ref{lemma:Aisbound}, we have:
  $$A_{i,s} \leq (\new{8}e)^{\frac{2^{s-2} - 1}{2^{s-2}}} \frac{i^{(2^s-1)/2^{s-1}}}{N^{(2^{s-1}-1)/2^{s-1}}} \cdot \Pi_s + O\left( s^4 \cdot \Pi_s \cdot N^{-1/(2^s(2^s-2))}\right).
  $$

  Hence we can bound $\srsc{i}{k}$ for any $i,k$, by:
  \begin{align*}
    \srsc{i}{k} &\leq \frac{A_{i,s+1}^k}{k!} + O\left(2^{-(s+1)^2 \cdot \Pi_s \cdot N^{1/2^{s+1}}}\right) \\
            &\leq \left(\frac{e \cdot A_{i,s+1}}{k}\right)^k + O\left(2^{-(s+1)^2 \cdot \Pi_s \cdot N^{1/2^{s+1}}}\right) \\
            &\leq \left(\frac{e}{k} (\new{8}e)^{\frac{2^{s-1} - 1}{2^{s-1}}} \frac{i^{(2^{s+1}-1)/2^{s}}}{N^{(2^{s}-1)/2^{s}}} \cdot \Pi_{s+1} + \frac{e}{k} \cdot O\left((s+1)^4  \Pi_{s+1} \cdot N^{-1/(2^{s+1}(2^{s+1}-2))}\right)\right)^k \\
            & \quad  + O\left(2^{-(s+1)^2 \cdot \Pi_s \cdot N^{1/2^{s+1}}}\right) \\
            &\leq \left(\frac{e}{k} \cdot (\new{8}e)^{\frac{2^{s-1} - 1}{2^{s-1}}} \frac{i^{(2^{s+1}-1)/2^{s}}}{N^{(2^{s}-1)/2^{s}}} \cdot 4(s+1) + \frac{e}{k} \cdot O\left(4(s+1)^{5} \cdot N^{-1/(2^{s+1}(2^{s+1}-2))}\right)\right)^k \\
            & \quad + O\left(2^{- 4s(s+1)^2 \cdot N^{1/2^{s+1}}}\right),
  \end{align*}
  where the first inequality comes from Lemma \ref{lemma:krscgik}, the third inequality comes from Lemma \ref{lemma:Aisbound} and the last inequality comes from Proposition \ref{proposition:Pi}.

  If $i=o\left((s+1)^{-\frac{2^{s}}{2^{s+1}-1}}\cdot k^{\frac{2^{s}}{2^{s+1}-1}}\cdot N^{\frac{2^{s}-1}{2^{s+1}-1}}\right)$, then $\srsc{i}{k}=o(1)$.
  Hence if we want $\srsc{i}{k}$ to be constant, i.e.\ not $o(1)$, we must have $i = \Omega\left(s^{-\frac{2^{s}}{2^{s+1}-1}}\cdot k^{\frac{2^{s}}{2^{s+1}-1}} \cdot N^{\frac{2^{s}-1}{2^{s+1}-1}}\right)$.
\end{proof}

\section{The $(r,k)$\---subset cover problem}
\label{section:rksc}

In this section, we prove some upper and lower bounds on the $(r,k)$\---SC problem.
As far as we know, there is no quantum algorithm to find a $(r,k)$\---SC problem, except for \cite{subres}'s algorithm when $k=r$, and for the harder problem of finding a $k$\---RSC.
We first prove a lower bound on the $(1,k)$\---SC problem, then design new algorithms for finding a $(r,k)$\---SC.

\subsection{Lower bound on finding a $(1,k)$\---subset cover}
\label{section:bound1ksc}
In this subsection, we will prove a lower bound on the $(1,k)$\---SC problem.
We are given $k$ random functions $h_1,\dots,h_k$ such that for $i \in [1,k]$, $h_i: \mathcal{X} \rightarrow \mathcal{Y}$.
We write $N = |\mathcal{Y}|$ and for $x \in \mathcal{X}$, we write $H(x) = \left\{h_i(x) \vert i \in [1,k]\right\}$.
The goal of this subsection is to prove the following theorem.
\begin{theorem}
  \label{theorem:1ksc}
  Given $k$ random functions $h_1,\dots,h_k: \mathcal{X}\rightarrow\mathcal{Y}$ where $N=|\mathcal{Y}|$, a quantum algorithm needs to make $\Omega\left(C_k^{-1/5}\cdot N^{k/5}\right)$ queries to $h_1,\dots,h_k$ to find one (1,k)\---SC with constant probability, where $C_k = \sum_{j=2}^k \frac{k!}{(j-1)!}$. \end{theorem}
To prove Theorem \ref{theorem:1ksc}, we introduce the problem of finding a \emph{$j$\---repetition} on $h_{i_1},\dots,h_{i_j}$, that consists in finding an $x \in \mathcal{X}$ such that $h_{i_1}(x) = \dots = h_{i_j}(x)$.
More formally, we define the following database property:
\begin{definition}
  $$
  \forall \ell,j, P_{\ell,j}^{rep} = \left\{D \in \mathcal{D} \middle\vert \begin{array}{l}
  \exists x_{1},x_{2},\dots,x_{\ell}, \forall i, \forall 1 \leq \ell \leq j, h_{1}(x_i) = h_{\ell}(x_i) \\
  \forall i \neq p, x_{i} \neq x_{p}
  \end{array}\right\}.
  $$
\end{definition}

Note that we define the property only for $\ell$ distinct $j$\---repetition on $h_1,\dots,h_j$, because by symmetry, the probability of finding a $j$\---repetition on $h_1,\dots,h_j$ is the same as finding a $j$\---repetition on $h_{i_1},\dots,h_{i_\ell}$.

We also define:
\begin{enumerate}
\item $\aljrepet{i}{\ell}{j}$ as the amplitude of the databases $D$ containing \emph{at least} $\ell$ distinct $j$\---repetitions on $h_1,\dots,h_j$ after $i$ quantum queries.
\item $\ljrepet{i}{\ell}{j}$ as the amplitude of the databases $D$ containing \emph{exactly} $\ell$ distinct $j$\---repetitions on $h_1,\dots,h_j$ after $i$ quantum queries.
\item $\ksc{i}{k}$ as the amplitude of the databases $D$ containing at least one $(1,k)$\---SC after $i$ quantum queries.
\end{enumerate}

More formally, let $\ket{\stateiafter}$ be the state just after the $i^{th}$ query to the oracle, then $\aljrepet{i}{\ell}{j} = \left|P_{\ell,j}^{rep}\ket{\stateiafter}\right|$, $\ljrepet{i}{\ell}{j} = \left|P_{\ell,j}^{rep}\neg P_{\ell+1,j}^{rep}\ket{\stateiafter}\right|$, and $\ksc{i}{k} = \left|P_{(1,k)}^{SC}\ket{\stateiafter}\right|$.

Our goal is to bound $\ksc{i}{k}$ and for that we will bound $\aljrepet{i}{\ell}{j}$.
\begin{lemma}
  \label{lemma:jrepet}
  For all $i,\ell,j \in \mathbb{N}$, we have that:
  $$
  \aljrepet{i}{\ell}{j} \leq \left(\frac{\new{4}e \cdot i}{\ell \cdot N^{\frac{j-1}{2}}}\right)^\ell.
  $$
\end{lemma}

\begin{proof}
  Following the proof of Lemma \ref{lemma:reck2rsc}, we have that:
  \begin{align*}
    \aljrepet{i}{\ell}{j} &\leq \aljrepet{i-1}{\ell}{j} + \sqrt{\frac{1}{N^{j-1}} \aljrepet{i-1}{\ell-1}{k}^2} \new{+ \frac{3 (i-1)}{N^{j}} \aljrepet{i-1}{\ell-1}{k}} \\
                          &\new{\leq \aljrepet{i-1}{\ell}{j} + 4 \sqrt{\frac{1}{N^{j-1}} \aljrepet{i-1}{\ell-1}{k}^2}} \\
              &\leq \sum_{m=0}^{i-1} \new{4}\sqrt{\frac{1}{N^{j-1}}} \aljrepet{m}{\ell-1}{k} \\
              &\leq \sum_{m_1=0}^{i-1}\sum_{m_2=0}^{m_1} \new{4}\sqrt{\frac{1}{N^{j-1}}}\new{4}\sqrt{\frac{1}{N^{j-1}}}\aljrepet{m_2}{\ell-2}{k}\\
              &\vdots \\
              &\leq \sum_{0 \leq m_\ell < m_{\ell-1} < \dots < m_1 < i} \left(\frac{\new{16}}{N^{j-1}}\right)^{\ell/2}\\
              &\leq \frac{i^\ell}{\ell!} \left(\frac{\new{16}}{N^{j-1}}\right)^{\ell/2}\\
              &\leq \left(\frac{\new{4}e \cdot i}{\ell \cdot N^{(j-1)/2}}\right)^{\ell},
  \end{align*}
\new{where the second inequality comes from the fact that we can assume $i \leq N^{j/2}$.}
\end{proof}

We now bound the amplitude $\ksc{i}{k}$ with an inductive formula, as for the RSC problem.
\begin{lemma}
  \label{lemma:1kscrec}
  For all $i \in \mathbb{N}$ and $k \in \mathbb{N}$, we have that:
  $$
  \ksc{i}{k} \leq \ksc{i-1}{k} + \new{4}\left(k^k\frac{i-1}{N^k}\right)^{1/2} + \left(\sum_{j=2}^k\sum_{\ell \geq 0} \frac{\ell}{N^{k+1-j}}\cdot\frac{k!}{(j-1)!}\ljrepet{i-1}{\ell}{j}^2\right)^{1/2}.
  $$
\end{lemma}

\begin{proof}
  For convenience, we denote $P_{k}=P_{(1,k)}^{SC}$ the projector on the databases $D$ that contain at least a $(1,k)$\---SC.
  We write $\ket{\stateibefore}$ the state just before the $i^{th}$ quantum query, and $\ket{\stateiafter}$ the state just after the $i^{th}$ quantum query.

  Using Lemma \ref{lemma:recursive}, and writing $D_{y'} := D \cup (x,y')$ we have that:
  \begin{align}
    \label{eq:1ksc}\left|P_{k}\ket{\stateiafter}\right| &\leq \left|P_{k}\ket{\stateibefore}\right|
    + \left|P_{k} \sum_{\substack{x,\hat{y},z \\ D : \text{no (1,k)\---SC}}} \frac{1}{\sqrt{N^k}} \sum_{y'} \omega_N^{yy'} \alpha_{x,\hat{y},z,D} \ket{x,\hat{y},z,D_{y'}}\right|\new{+ 3 \frac{i-1}{N^k}}
    \end{align}

  We analyse now the possibilities for achieving a $(1,k)$\---SC, considering the different cases of the inner sum.
  We have multiple possible ways to get from $D$ that does not have a $(1,k)$\---SC to $D_{y'}$ that has a $(1,k)$\---SC.
  \begin{itemize}
  \item ($x=x_0$) Here, we consider the case where we query $x$ such that $\{h_i(x)\} \subseteq \{h_i(x_1)\}$, where $x_1$ was queried before. Notice that there are $(i-1)$ possible values of $x_1$, and for each fixed value of $x_1$, we have $k^k$ possible values of $H(x)$ that would lead to this value. This leads to $k^k (i-1)$ possible values of $y'$ that would lead to an $(1,k)$-SC.
  \item $(x = x_1)$ Here, we consider the case where we query $x$ such that $\{h_i(x_0)\} \subseteq \{h_i(x)\}$, where $x_0$ was queried before.  
  
  Let us suppose that $x_0$ has a $j$-repetition on $h_{i_1},\dots,h_{i_j}$, for some distinct $i_1,...,i_j$. Notice that in this case, $S:=\{h_i(x_0)\}$ has $k-j+1$ elements and we will count the number of possible $H(x)$ that contains all of these elements.
  Out of the $k$ functions $h_1,\dots,h_k$, we have $\binom{k}{k-j+1}$ possible ways of choosing the functions that will be filled with the values in $S$. When we fix such functions, there are $|S|!=(k-j+1)!$ ways of filling them with the elements of $S$, and $N^{j-1}$ ways of filling the other functions. Therefore, there are $\binom{k}{k-j+1}(k-j+1)!N^{j-1}$ values of $H(x)$ such that $\{h_i(x_0)\} \subseteq \{h_i(x)\}$. \\

  \end{itemize}
  This gives, \new{bounding the last term of \Cref{eq:1ksc} by $3\left(k^k\frac{i-1}{N^k}\right)^{1/2}$}: %
  \begin{align*}
    \left|P_{k}\ket{\stateiafter}\right| 
    & \leq \left|P_{k}\ket{\stateibefore}\right|
      + \new{4}\left(k^k\frac{i-1}{N^k}\sum_{\substack{x,\hat{y},z \\ D : \text{no k\---SC}}}\left|\alpha_{x,\hat{y},z,D}\right|^2\right)^{1/2} \\
    &+ \left(\sum_{j=2}^k\sum_{\ell \geq 0} \frac{\ell}{N^{k+1-j}}\cdot\frac{k!}{(j-1)!}\sum_{\substack{x,\hat{y},z \\ D : \text{no k\---SC} \\ \ell \text{ distinct } j-repetitions}}\left|\alpha_{x,\hat{y},z,D}\right|^2\right)^{1/2}.
  \end{align*}

  Using Lemma \ref{lemma:commute} and our notations, we conclude:
  $$
  \ksc{i}{k} \leq \ksc{i-1}{k} + \new{4}\left(k^k\frac{i-1}{N^k}\right)^{1/2} + \left(\sum_{j=2}^k\sum_{\ell \geq 0} \frac{\ell}{N^{k+1-j}}\cdot\frac{k!}{(j-1)!}\ljrepet{i-1}{\ell}{j}^2\right)^{1/2}.
  $$
\end{proof}

We now bound $\ksc{i}{k}$ in the following lemma.
\begin{lemma}
  \label{lemma:1kscbound}
  For every $i \in \mathbb{N}$ and $k \in \mathbb{N}$, we have that:
  $$
  \ksc{i}{k} \leq \new{4}k^{k/2}\cdot\frac{i^{3/2}}{N^{k/2}}
  + \sqrt{\sum_{j=2}^{k}\frac{k!}{(j-1)!}}\cdot\frac{\new{4}e \cdot  i^{5/2}}{N^{k/2}}.
  $$
\end{lemma}

\begin{proof}
  From Lemma \ref{lemma:1kscrec}, we have that:
  $$
  \ksc{i}{k} \leq \ksc{i-1}{k} + \new{4}\left(k^k\frac{i-1}{N^k}\right)^{1/2} + \left(\sum_{j=2}^k\sum_{\ell \geq 0} \frac{\ell}{N^{k+1-j}}\cdot\frac{k!}{(j-1)!}\ljrepet{i-1}{\ell}{j}^2\right)^{1/2}.
  $$
  
  We want to bound each term in the sum indexed by $j$.
  Fix $j \in \{2,\dots,k\}$. We have that:
  \begin{align*}
    \sum_{\ell\geq0}\frac{\ell}{N^{k+1-j}}\cdot \frac{k!}{(j-1)!} \ljrepet{i-1}{\ell}{j}^2 = \frac{k!}{(j-1)!}\cdot\sum_{\ell\geq0}\frac{\ell}{N^{k+1-j}}\ljrepet{i-1}{\ell}{j}^2.
  \end{align*}

  Next, we have that:
  \begin{align*}
    \sum_{\ell\geq0}\frac{\ell}{N^{k+1-j}}\ljrepet{i-1}{\ell}{j}^2 &\leq \frac{i-1}{N^{k+1-j}} \cdot \sum_{\ell\geq1}\ljrepet{i-1}{\ell}{j}^2 \\
                                                           &= \frac{i-1}{N^{k+1-j}} \cdot \aljrepet{i-1}{1}{j}^{2}\\
                                                           &\leq \frac{i-1}{N^{k+1-j}} \cdot \left(\frac{\new{4}e \cdot (i-1)}{N^{\frac{j-1}{2}}}\right)^2 \\
                                                           &= \frac{(\new{4}e)^2  (i-1)^{3}}{N^{k}},
  \end{align*}
  where $\aljrepet{i-1}{1}{j}$ is the amplitude of the databases $D$ containing \emph{at least} one $j$\---repetition on $h_1,\dots,h_j$ after $i-1$ quantum queries.
  The first inequality follows since there cannot be more than $i-1$ distinct $j$\---repetitions on $h_1,\dots,h_j$ after $i-1$ quantum queries.
  The second inequality comes from the bound on $\aljrepet{i-1}{1}{j}$ in Lemma \ref{lemma:jrepet}.

  This gives:
  \begin{align*}
    \left(\sum_{j=2}^k\frac{k!}{(j-1)!}\sum_{\ell\geq0}\frac{\ell}{N^{k+1-j}}\ljrepet{i-1}{\ell}{j}^2\right)^{1/2} &\leq \sqrt{\sum_{j=2}^k\frac{k!}{(j-1)!}} \cdot \frac{\new{4}e \cdot  (i-1)^{3/2}}{N^{k/2}}.
  \end{align*}

  Finally, by developing the recursive terms (and using that $\ksc{0}{k}=0$), we get that:
  \begin{align*}
    \ksc{i}{k} &\leq \ksc{i-1}{k} + \new{4}\sqrt{k^k\frac{i-1}{N^k}}+ \sqrt{\sum_{j=2}^k\frac{k!}{(j-1)!}} \cdot \frac{\new{4}e\cdot (i-1)^{3/2}}{N^{k/2}} \\
            &\vdots \\
            &\leq \sum_{\ell=0}^{i-1} \left(\new{4}\sqrt{k^k\frac{\ell}{N^k}}+ \sqrt{\sum_{j=2}^k\frac{k!}{(j-1)!}} \cdot \frac{\new{4}e\cdot \ell^{3/2}}{N^{k/2}}\right) \\
            &\leq \new{4}k^{k/2}\frac{i^{3/2}}{N^{k/2}}
              + \sqrt{\sum_{j=2}^{k}\frac{k!}{(j-1)!}}\cdot\frac{\new{4}e \cdot  i^{5/2}}{N^{k/2}}.
  \end{align*}
\end{proof}

We can now prove Theorem \ref{theorem:1ksc}.
\begin{proof}[Proof of Theorem \ref{theorem:1ksc}]
  From Lemma \ref{lemma:1kscbound}, we have that:
  $$
  \ksc{i}{k} \leq \new{4}k^{k/2}\cdot\frac{i^{3/2}}{N^{k/2}}
  + \sqrt{\sum_{j=2}^{k}\frac{k!}{(j-1)!}}\cdot\frac{\new{4}e \cdot  i^{5/2}}{N^{k/2}}.
  $$

  Writing $C_k= \sum_{j=2}^{k}\frac{k!}{(j-1)!}$, this rewrites as:
  $$
  \ksc{i}{k} \leq \new{4}k^{k/2}\cdot\frac{i^{3/2}}{N^{k/2}}
  + \sqrt{C_k}\cdot\frac{\new{4}e \cdot  i^{5/2}}{N^{k/2}}.
  $$

  If $i=o\left(C_k^{-1/5}\cdot N^{k/5}\right)$, then $\ksc{i}{k}=o(1)$.
  Hence if we want $\ksc{i}{k}$ to be constant, i.e.\ not $o(1)$, we must have $i=\Omega\left(C_k^{-1/5}\cdot N^{k/5}\right)$.
\end{proof}

\subsection{Algorithm for finding a $(1,k)$\---subset cover}
\label{section:algo1ksc}

We now describe an algorithm that finds a $(1,k)$\---SC, assuming
 $|\mathcal{X}| = |\mathcal{Y}|^k = N^k$. 
We first notice that an algorithm that finds a collision on $H$ also finds a $(1,k)$\---SC in an expected $O(N^{k/3})$ number of queries.
We show now that there is a more efficient algorithm, as stated in the following theorem:
\begin{theorem}
  \label{theorem:algo1ksc}
  There exists a quantum algorithm that finds a $(1,k)$\---SC in expected $O\left(N^{k/4}\right)$ quantum queries if $k$ is even, and $O(N^{k/4 + 1/12})$ if $k$ is odd.
\end{theorem}

To prove this theorem, we describe the following algorithm (which takes as parameters $j$ and $t$, whose values will be chosen later):
\begin{algorithm} Input: $j \in \{2,\dots,k\}$ and $t \in \mathbb{N}$.
  \label{algo:1ksc}
  \begin{enumerate}
  \item Define $F_1:\mathcal{X} \rightarrow \{0,1\}$ as follows:
    $$
    F_1(x) =
    \begin{cases}
      1,& \text{if } h_{1}(x) = h_{2}(x) = \dots = h_{j}(x) \\
      0,& \text{otherwise.}
    \end{cases}
    $$
    (Note that an element $x \in \mathcal{X}$ such that $F_1(x) = 1$ is a \emph{$j$\---repetition}.)
  \item Execute Grover's algorithm $t$ times on $F_1$ to find $t$ distinct $j$\---repetitions in $H$.
    Let T = $\{x_1,\dots,x_t\}$ be the set of these $j$\---repetitions.
  \item Define $F_2:\mathcal{X} \rightarrow \{0,1\}$ as follows:
    $$
    F_2(x) =
    \begin{cases}
      1,& \text{if there exists } x_0 \in T \text{ such that }h_{1}(x) = h_{1}(x_0) \\
        & \text{and for } 1 \leq m \leq k-j, h_{m+1}(x) = h_{j+m}(x_0) \\
      0,&\text{otherwise.}
    \end{cases}
    $$
  \item Execute Grover's algorithm to find an $x$ such that $F_2(x)=1$
  \item Find $x_0$ in $T$ corresponding to $x$, and output $(x,x_0)$.
  \end{enumerate}
\end{algorithm}

\begin{lemma}
  \label{lemma:algo1kscgeneral}
  Algorithm \ref{algo:1ksc} makes an expected number of $O\left(N^{(2k-j+1)/6}\right)$ queries to the oracle when $j \leq \frac{k+2}{2}$ for $t=N^{(k-2j+2)/3}$.
\end{lemma}
\begin{proof}

  Notice that if we consider a uniformly random function, we have that $
  Pr[h_1(x) = \cdots = h_j(x)] = N^{-j+1}$. Therefore, the expected number of elements in $\mathcal{X}$ such that 
  $F_1(x) = 1$ is $N^k \cdot N^{-j+1} =  N^{k-j+1}$.
  We write $X_1,\dots,X_{N^k}$ the random variables corresponding to $F_1$'s output on each $x \in \mathcal{X}$, $X$ the sum of these variables, $\mu=N^{k-j+1}$ their mean.
  Chernoff bound tells us that for any $0\leq\delta \leq1$,
  $$
  Pr\left(\vert X - \mu \vert \geq \mu\delta \right) \leq e^{-\delta^2\mu/3}.
  $$
  With $\delta=1/2$, we have:
  $$
  Pr\left(\vert X - \mu \vert \geq \frac{\mu}{2} \right) \leq e^{-\mu/12}.
  $$
  Thus, unless with probability $e^{-(N^{k-j+1})/12}$, the number of elements $x \in \mathcal{X}$ such that $F_1(x)=1$ is greater than $N^{k-j+1}/2$.

  Hence using Theorem \ref{theorem:grover}, the second step of the algorithm is expected to make $O\left(t \cdot \sqrt{\frac{N^k}{N^{k-j+1}}}\right) = O\left(t \cdot N^{(j-1)/2}\right)$ quantum queries to the oracle.

  Notice that for a fixed value $x_0$, if we consider a uniformly random function, we have that 
  $$Pr[h_{1}(x) = h_{1}(x_0) \land h_2(x) = h_{m+1}(x_0) \land \dots \land h_{k-j+1}(x) = h_k(x)] = N^{j-k-1}.$$ 
  Therefore, the expected number of elements such that 
  $F_2(x) = 1$ is $t \cdot N^k \cdot N^{j-k-1} =  t \cdot N^{j-1}$.
  Similarly, using Chernoff bound, unless with probability $e^{-(t\cdot N^{j-1})/12}$, the number of elements such that $F_2(x) = 1$ is greater than $t \cdot N^{j-1}/2$.
  Hence, using Theorem \ref{theorem:grover}, the fourth step of the algorithm is expected to make $O\left(\sqrt{\frac{N^k}{t \cdot N^{j-1}}}\right) = O\left(\frac{N^{(k-j+1)/2}}{\sqrt{t}}\right)$ quantum queries to the oracle.

  By picking $t=N^{(k-2j+2)/3}$ with $j\leq\frac{k+2}{2}$ (otherwise $t<1$), the complexity of the algorithm is $O(N^{(k-2j+2)/3} \cdot N^{(j-1)/2}) = O(N^{(2k-j+1)/6})$.
\end{proof}

We now prove Theorem \ref{theorem:algo1ksc}

\begin{proof}[Proof of Theorem \ref{theorem:algo1ksc}]
  From Lemma \ref{lemma:algo1kscgeneral}, the complexity of Algorithm \ref{algo:1ksc} is \ifx\format\plain\else\linebreak\fi
  $O(N^{(2k-j+1)/6})$ when $j \leq \frac{k+2}{2}$.
  \begin{itemize}
  \item If $k$ is even, then we pick $j=\frac{k+2}{2}$ to reach a complexity of $O(N^{k/4})$.
  \item If $k$ is odd, then we pick $j=\frac{k+1}{2}$ to reach a complexity of $O(N^{k/4+1/12})$.
  \end{itemize}

  Note that if $j > \frac{k+1}{2}$, then the second step of the algorithm is expected to make at least $O\left(N^{\frac{k+1}{4}}\right)$ quantum queries, which is worse than $O(N^{k/4+1/12})$.
\end{proof}

\begin{rem}
Note that we do not reach the lower bound of Theorem \ref{theorem:1ksc}, and it would be interesting to see if the gap can be further reduced by either improving our lower bounds or designing a more efficient algorithm.
\end{rem}

\paragraph{A slightly better algorithm}
We describe a more efficient algorithm when $k$ is not constant.
The idea is to take into account the fact that we do not necessarily need the $j$\---repetitions from the previous algorithm to occur on the first $j$ functions $h_1,\dots,h_j$, but they could rather be on any $h_{i_1},\dots,h_{i_j}$ instead.
We also consider permutations of the $h_1,\dots,h_k$ in the fourth step of Algorithm \ref{algo:1ksc}.

\begin{theorem}
  \label{theorem:algo1kscopti}
  There exists a quantum algorithm that finds a $(1,k)$\---SC in:
  \begin{itemize}
      \item $O\left(\binom{k}{(k+2)/2}^{-1/2} \cdot N^{k/4}\right)$ quantum queries if $k$ is even, 
      \item $O\left(\binom{k}{(k+1)/2}^{-1/2} \cdot N^{k/4+1/12}\right)$ quantum queries if $k$ is odd.
  \end{itemize}
\end{theorem}
The gain that we obtain is a function of $k$ and is therefore not significant if $k$ is constant. However, as we have shown in Theorem \ref{theorem:1ksc}, the dependence in $k$ can be quite large for the $(1,k)$\---SC problem.

To prove this theorem, we describe the algorithm as follows (which takes again as input two integers $j$ and $t$ playing the role of parameters whose optimal values will be determined later):
\begin{algorithm}
  \label{algo:1kscopti}
  Input: $j\in \{2,\dots,k\}$ and $t \in \mathbb{N}$.
  \begin{enumerate}
  \item Define $F_1:\mathcal{X} \rightarrow \{0,1\}$ as follows:
    $$
    F_1(x) = 
    \begin{cases}
      1,& \text{if there exists distinct }i_1,\dots,i_j \in [1,k] \text{ such that} \\
        & h_{i_1}(x) = h_{i_2}(x) = \dots = h_{i_j}(x) \\
      0,& \text{otherwise.}
    \end{cases}
    $$
    (Note that an element $x \in \mathcal{X}$ such that $F_1(x) = 1$ is a \emph{$j$\---repetition}.)
  \item Execute Grover's algorithm $t$ times on $F_1$ to find $t$ distinct j\---repetitions in $H$.
    Let T = $\{x_1,\dots,x_t\}$ be the set of these j\---repetitions.
    We write, for $\ell \in [1,t]$ $I_\ell = \{i_1^{\ell},\dots,i_j^{\ell}\}$ the set of indices such that $h_{i_1^{\ell}}(x_t) = \dots = h_{i_k^{\ell}}(x_t)$, and $I'_\ell = [1,k]\backslash I_\ell = \{i_{j+1}^{\ell},\dots,i_{k}^{\ell}\}$. 
  \item Define $F_2:\mathcal{X} \rightarrow \{0,1\}$ as follows:
    $$
    F_2(x) = 
    \begin{cases}
      1, & \text{if there exists distinct }j_0,j_1,\dots,j_{k-j+1} \in [1,k], \\
      & \text{ and } \ell \in [1,t]  %
       \text{ s.t. } h_{i_1^{\ell}}(x_\ell) = h_{j_0}(x) \\
&      \text{ and for all }  1 \leq m \leq k-j, h_{j_m}(x) = h_{i_{j+m}^{\ell}}(x_\ell) \\
      0,& \text{otherwise.}
    \end{cases}
    $$
  \item Execute Grover's algorithm to find an $x$ such that $F_2(x) = 1$
  \item Find $x_0$ in $T$, and output $(x,x_0)$.
  \end{enumerate}
\end{algorithm}

\begin{rem} $F_1$ (resp. $F_2$) can be constructed with $O\left(\binom{k}{j}\right)$ (resp. $O\left(\frac{k!}{(j-1)!}\right)$) quantum gates and one query to $H$.
\end{rem}

\begin{lemma}
  \label{lemma:algo1kscoptigeneral}
  Algorithm~\ref{algo:1kscopti} makes an expected number of  $O\left(\binom{k}{j}^{-1/2}N^{(2k-j+1)/6}\right)$ queries to the oracle when $j \leq \frac{k+2}{2}$ for $t=N^{(k-2j+2)/3}$.
\end{lemma}
The proof of Lemma \ref{lemma:algo1kscoptigeneral} is given in \Cref{sec:algo1kscoptigeneral}. We now prove Theorem~\ref{theorem:algo1kscopti}.
\begin{proof}[Proof of Theorem \ref{theorem:algo1kscopti}]
  From  Lemma~\ref{lemma:algo1kscoptigeneral}, the complexity of Algorithm \ref{algo:1kscopti} is \ifx\format\plain\else\linebreak\fi $O\left(\binom{k}{j}^{-1/2} \cdot N^{(2k-j+1)/6}\right)$ when $j \leq \frac{k+2}{2}$.
\begin{itemize}
\item If $k$ is even, for $j=\frac{k+2}{2}$, we get a complexity of $O\left(\binom{k}{(k+2)/2}^{-1/2} \cdot N^{k/4}\right)$.
\item If $k$ is odd, for  $j=\frac{k+1}{2}$, we get a complexity of $O\left(\binom{k}{(k+1)/2}^{-1/2} \cdot N^{k/4+1/12}\right)$.
\end{itemize}

Note that if $j > \frac{k+1}{2}$, then the second step of the algorithm is expected to make at least $O\left(\binom{k}{(k+1)/2}^{-1/2} \cdot N^{\frac{k+1}{4}}\right)$ quantum queries.
\end{proof}

\subsection{Algorithm for finding a $(r,k)$\---subset cover}
\label{section:algorksc}

In this section, we describe an algorithm for solving the $(r,k)$\---SC problem.
We consider the case where $|\mathcal{X}| = |r\cdot\mathcal{Y}|^k = r^k\cdot N^k$.
The result is stated as follows:
\begin{theorem}
  \label{theorem:rkscalgo}
  There exists a quantum algorithm that finds a $(r,k)$\---SC in \ifx\format\plain\else\linebreak\fi $O\left(N^{k/(2+2r)}\right)$ quantum queries to $H$, if $k$ is divisible by $r+1$, and \ifx\format\plain\else\linebreak\fi $O\left(N^{k/(2+2r) + 1/2}\right)$ otherwise.
\end{theorem}

The idea of the algorithm is essentially the same as Algorithm \ref{algo:1ksc} of Section \ref{section:algo1ksc}:
\begin{enumerate}
\item we first find $t$ distinct $(r-1,k')$\---SC for some integers $t$ and $k'$;
\item we then find the $(r,k)$\---SC.
\end{enumerate}

The first step is done recursively, using the algorithm defined for lower values of $k'$ and $r-1$.
The second step uses Grover's algorithm.
The algorithm can be defined for any value of $k'$ and $t$, and we pick them to optimize the complexity.

More formally, we define the algorithm recursively.
Assume that we have an algorithm that can output a $(r-1,k')$\---SC in $O\left(N^{k'/2r}\right)$ queries, for any $k'<k$ such that $k'$ is divisible by $r$.
Then, we can find a $(r,k)$\---SC as follows:

\begin{algorithm}
\label{algo:rksc}
Input: $t \in \mathbb{N}$, $k' \in \mathbb{N}$.
\begin{enumerate}
\item Execute the $(r-1,k')$\---SC algorithm $t$ times to find $t$ distinct $(r-1,k')$\---SC in $H$.
  Let T = $\{(x_{1,0},x_{1,1},\dots,x_{1,r-1}),\dots,(x_{t,0},x_{t,1},\dots,x_{t,r-1})\}$ be the set of these $(r-1,k')$\---SC.
\item Define $F:\mathcal{X} \rightarrow \{0,1\}$ as follows:
  $$
  F(x) =
  \begin{cases}
    1, & \text{if there exists } (x_{i,0},x_{i,1},\dots,x_{i,r-1}) \in T \text{ such that} \\
      &  \forall 1 \leq m \leq k-k', h_{m}(x) = h_{k'+m}(x_{i,0}), \\
    0, & \text{ otherwise.}
  \end{cases}
  $$
\item Execute Grover's algorithm to find an $x$ such that $F(x) = 1$
\item Find $(x_{i,0},x_{i,1},\dots,x_{i,r-1})$ in $T$ and output $(x_{i,0},x_{i,1},\dots,x_{i,r-1},x)$.
\end{enumerate}
\end{algorithm}

\begin{lemma}
  \label{lemma:rkscalgogeneral}
  Algorithm \ref{algo:rksc} makes an expected number of $O\left(N^{k/(2+2r)}\right)$ queries to the oracle, when $k$ is divisible by $r$, and $O\left(N^{k/(2+2r)+1/2}\right)$ otherwise.
\end{lemma}

We defer the proof of Lemma~\ref{lemma:rkscalgogeneral} to \Cref{sec:rkscalgogeneral}.

\section*{Acknlowedgements}
ABG is supported by ANR JCJC TCS-NISQ ANR-22-CE47-0004, and by the PEPR integrated project EPiQ ANR-22-PETQ-0007 part of Plan France 2030. 
This work is part of HQI initiative (www.hqi.fr) and is supported by France 2030 under the French National Research Agency award number “ANR-22-PNCQ-0002”.
We thanks the anonymous reviewers for their valuable comments that helped improving the quality of this paper.

\bibliographystyle{abbrv}
\bibliography{abbrev3,crypto,bibliography}

\clearpage
\appendix

\section{Grover's algorithm and Quantum Fourier Transform}
\label{sec:grover_qft}
\subsection{Grover's algorithm}

Here we quickly recall Grover's algorithm.
We start by defining the search problem.

\begin{definition}[Search problem]
  We are given a function $F : \mathcal{X} \rightarrow \{0,1\}$.
  The search problem consists of finding an $x \in \mathcal{X}$ such that $F(x) = 1$, in the least amount of queries to $F$ possible.
\end{definition}

Grover's algorithm solves the search problem in $O\left(\sqrt{\frac{|\mathcal{X}|}{t}}\right)$, where $t$ is the number of $x$ such that $F(x) = 1$.
The result is stated as follows:

\begin{theorem}[\cite{STOC:Grover96}\cite{BBHT}]
\label{theorem:grover}
Let $F : \mathcal{X} \rightarrow \{0,1\}$ be a function, $t = \left|\{x\vert F(x)=1\}\right|$, and $N = |\mathcal{X}|$.
  Then, Grover's algorithm finds an $x$ such that $F(x)=1$ with constant probability with $O\left(\sqrt{\frac{N}{t}}\right)$ queries to F. Moreover, this algorithm is optimal.
\end{theorem}

\begin{rem} When constructing quantum algorithms in the \emph{Quantum Random Oracle Model}, we are given a black box access to a function $H: \mathcal{X} \rightarrow \mathcal{Y}$.
To use Grover's algorithm in this model, we need to construct the function $F: \mathcal{X} \rightarrow \{0,1\}$ from the function $H$.
Then, to count the number of queries to $H$, it is sufficient to compute the number of queries to $F$.
\end{rem}

\subsection{The Quantum Fourier Transform}

Let $Y=\{0,1\}^n$, for some $n \in \mathbb{N}$.
We recall that the computational basis is $\{\ket{y}\}_{y \in Y}$.
The \textbf{Quantum Fourier Transform} is a unitary that, given an input state $\ket{\phi} = \sum_{k=0}^{2^n-1}x_k \ket{k}$, outputs $\sum_{k=0}^{2^n}y_k \ket{k}$ where the $y_k$'s are computed with the following formula:
$$
y_k = \frac{1}{2^{n/2}}\sum_{\ell=0}^{2^n-1}x_\ell\omega_N^{k\ell}
$$

where $\omega_N = e^{2\pi i/2^n}$ thus $\omega_N^\ell$ is a $2^n$-th root of unity.

This unitary can be efficiently implemented, and we write it \emph{QFT}.

Applying the $QFT$ to the computational basis yields the \emph{Fourier basis} $\{\ket{\hat{y}}\}_{y \in Y}$.

\section{Technical proofs}

\subsection{Proof of Lemma \ref{lemma:j2col}}
\label{appendix:j2col}
\new{As previously mentioned, the proof closely follows the proof of Corollary 11 from \cite{EC:LiuZha19}.

We write  $P'_{\ell-col-h_1}$ the set of databases that contain \emph{at least} $\ell$ \emph{distinct} collisions on $h_1$.

Let $i \in \mathbb{N}$. Let $\ket{\stateibefore}$ be the state just before the $i^{th}$ query to $H = (h_1,h_2)$, namely
\begin{align*}
  \ket{\stateibefore} = \sum_{x,\hat{y},z,D} \alpha_{x,\hat{y},z,D} \ket{x,\hat{y},z} \otimes \ket{D},
\end{align*}
where $x$ is the query register, $y$ is the answer register, $z$ is the work register and $D$ is the database register.
Let $\ket{\stateiafter}$ be the state right after the  $i^{th}$ query to H, namely
\begin{align*}
  \ket{\stateiafter} = \sum_{\substack{x,\hat{y},z,D \\ \new{D(x) = \bot}}} \frac{1}{\sqrt{N^2}}\sum_{y'} \omega_N^{yy'}\alpha_{x,\hat{y},z,D} \ket{x,\hat{y},z} \otimes \ket{D\cup(x,y')} \new{+ \cO \sum_{\substack{x,\hat{y},z,D \\ D(x) \neq \bot}} \alpha_{x,\hat{y},z,D} \ket{x,\hat{y},z} \otimes \ket{D}}.
\end{align*}

From Lemma \ref{lemma:recursive}, we have that:

\begin{equation}
  \left|P'_{\ell-col-h_1}\ket{\stateiafter}\right| \leq \left|P'_{\ell-col-h_1}\ket{\stateibefore}\right| + \left|P'_{\ell-col-h_1} \cO (I - P'_{\ell-col-h_1})\ket{\stateibefore}\right|.
\end{equation}

Writing $\acolh{i}{\ell} = \left|P'_{\ell-col-h_1}\ket{\stateiafter}\right|$ and similarly to the proof of Lemma \ref{lemma:rec2rsc2}, using Lemma \ref{lemma:matrixCO} and Lemma \ref{lemma:commute} we obtain the following recursive inequality:

\begin{align*}
\acolh{i}{\ell} &\leq \acolh{i-1}{\ell} + 4 \frac{\sqrt{i-1}}{\sqrt{N}} \acolh{i-1}{\ell-1}\\
&\leq \sum_{j=0}^{i-1} 4 \frac{\sqrt{j}}{\sqrt{N}}  \acolh{j}{\ell-1} \\  
&\leq \sum_{j_1=0}^{i-1} 4 \frac{\sqrt{j_1}}{\sqrt{N}} \sum_{j_2=0}^{j_1-1} 4 \frac{\sqrt{j_2}}{\sqrt{N}}  \acolh{j_2}{\ell-2} \\   
&\vdots \\
&\leq \sum_{0 \leq j_{\ell} < j_{\ell-1} < \dots < j_1 < i} \prod_{k=1}^{\ell} 4\frac{\sqrt{j_k}}{\sqrt{N}} \\  
&\leq \frac{1}{\ell!} \sum_{0 \leq j_{\ell}, j_{\ell-1}, \dots, j_1 < i} \prod_{k=1}^{\ell} 4\frac{\sqrt{j_k}} {\sqrt{N}} \\
&= \frac{1}{\ell!} \left(\sum_{0 < j < i} 4 \frac{\sqrt{i-1}}{\sqrt{N}}\right)^j\\
&\leq \left(\frac{4e \cdot i^{3/2}}{j\sqrt{N}}\right)^j,
\end{align*}
where the computation follows from the proof of Lemma 11 in \cite{EC:LiuZha19}.
}
\subsection{Proof of \Cref{eq:projection-2rsc-app}}
\label{sec:projection-2rsc-app}
Writing $D_{y'} = D \cup (x,y')$,
\begin{align*}
  &\left| P_2 \sum_{y'} \frac{1}{\sqrt{N^2}}
   \sum_{\substack{x,\hat{y},z \\ D : \neg P_2 \\ D(x) = \bot}}
   \omega_N^{yy'}\alpha_{x,\hat{y},z,D}\ket{x,\hat{y},z,D_{y'}} \right|\\      
   \leq& \left| \sum_{\ell\geq0} \frac{\ell}{N}
   \sum_{b \in \{1,2\}}
   \sum_{\substack{x,\hat{y},z \\ D : \neg P_2 \\ \text{exactly $\ell$} \\ \text{collisions on $h_b$}}}
   \sum_{y'} \frac{1}{\sqrt{N^2}} \omega_N^{yy'}\alpha_{x,\hat{y},z,D}\ket{x,\hat{y},z,D_{y'}}\right.\\
      &+ \left.\frac{(i-1)^2}{N^2}\sum_{\substack{x,\hat{y},z \\ D : \neg P_2}}
      \sum_{y'} \frac{1}{\sqrt{N^2}} \omega_N^{yy'}\alpha_{x,\hat{y},z,D}\ket{x,\hat{y},z,D_{y'}}\right| \\
  \leq& \left|2 \cdot \sum_{\ell\geq0} \frac{\ell}{N}\sum_{\substack{x,\hat{y},z \\ D : \neg P_2 \\ \text{exactly $\ell$} \\ \text{collisions on $h_1$}}}
  \sum_{y'} \frac{1}{\sqrt{N^2}} \omega_N^{yy'}\alpha_{x,\hat{y},z,D}\ket{x,\hat{y},z,D_{y'}} \right| \\
      &+ \left|\frac{(i-1)^2}{N^2}\sum_{\substack{x,\hat{y},z \\ D : \neg P_2}}
      \sum_{y'} \frac{1}{\sqrt{N^2}} \omega_N^{yy'}\alpha_{x,\hat{y},z,D}\ket{x,\hat{y},z,D_{y'}}\right| \\
      \leq& \left(2\cdot\sum_{\ell\geq0} \frac{\ell}{N}\sum_{\substack{x,\hat{y},z \\ D : \neg P_2 \\ \text{exactly $\ell$} \\ \text{collisions on $h_1$}}} \left|\alpha_{x,\hat{y},z,D}\right|^2\right)^{1/2}
      + \left(\frac{(i-1)^2}{N^2}\sum_{\substack{x,\hat{y},z \\ D : \neg P_2}} \left|\alpha_{x,\hat{y},z,D}\right|^2\right)^{1/2}  \\
  \leq& \left(2 \cdot \sum_{\ell\geq0} \frac{\ell}{N} \left|P_{\ell-col-h_1}\ket{\stateibefore}\right|^2\right)^{1/2}
    + \frac{(i-1)}{N}, \\
\end{align*}
where in the second inequality, we used the symmetry of finding collisions on $h_1$ and collisions on $h_2$, and used the definition of $\left|P_{\ell-col-h_1}\ket{\stateibefore}\right|^2$ in the last inequality.

\subsection{Proof of Lemma~\ref{lemma:boundAi}}
\label{sec:boundAi}
\begin{proof}
  We have that
  \begin{align*}
    A_i &\leq \sum_{\ell:\mu(\ell)={\new{8}e\frac{\ell^{3/2}}{\sqrt{N}}}}\sqrt{2} \cdot \frac{\sqrt{\new{8}e\ell^{3/2}}}{N^{3/4}} + \sum_{\ell:\mu_3(\ell)=10N^{1/8}} \sqrt{2} \cdot \frac{\sqrt{10N^{1/8}}}{N^{1/2}} + \sum_{\ell=0}^{i-1}\new{4} \cdot \frac{\ell-1}{N}\\
        &\leq \sum_{\ell=1}^{i-1}\sqrt{2} \cdot \frac{\sqrt{\new{8}e\ell^{3/2}}}{N^{3/4}} + \sum_{\ell:\mu_3(\ell)=10N^{1/8}} \sqrt{2} \cdot \frac{\sqrt{10N^{1/8}}}{N^{1/2}} + \sum_{\ell=0}^{i-1}\new{4} \cdot \frac{\ell-1}{N}\\
        &\leq \new{4}\sqrt{e}\frac{i^{7/4}}{N^{3/4}} + \sqrt{2}\cdot\left(\frac{10}{\new{8}e}\right)^{2/3}\cdot N^{5/12} \cdot \frac{\sqrt{10N^{1/8}}}{N^{1/2}} + \new{4}\cdot\frac{i^2}{N} \\
        &\leq \new{4}\sqrt{e}\cdot\frac{i^{7/4}}{N^{3/4}} + \new{4}\cdot\frac{i^2}{N} + O\left(N^{-1/48}\right), \\
  \end{align*}
  where the third inequality comes from counting the number of $\ell$ such that $\mu_3(\ell)=10N^{1/8}$, which is equal to the number of $\ell$ such that $\new{8}e\frac{\ell^{3/2}}{\sqrt{N}} \leq 10N^{1/8}$.
\end{proof}

\subsection{Proof of \Cref{equation:bound2sc}}
\label{appendix:recursion}
Here, we give a proof of \Cref{equation:bound2sc}.
Starting from \Cref{equation:beforebound2sc}, we have that:

\begin{align*}
  \ktworsc{i}{k} &\leq \ktworsc{i-1}{k} + \sqrt{2}\left(\sqrt{\frac{\mu_3(i-1)}{N}} + \new{\sqrt{8}}\frac{i-1}{N}\right)\ktworsc{i-1}{k-1} + \sqrt{2}\cdot \acolh{i-1}{\mu_3(i-1)} \\
  &\vdots \\
          &\leq \sqrt{2}\sum_{\ell=0}^{i-1} \left(\left(\sqrt{\frac{\mu_3(\ell)}{N}} + \new{\sqrt{8}}\frac{\ell}{N}\right)\ktworsc{\ell}{k-1} + \acolh{\ell}{\mu_3(\ell)}\right) \\
          &\leq \sqrt{2}\sum_{\ell=0}^{i-1} \left(\left(\sqrt{\frac{\mu_3(\ell)}{N}} + \new{\sqrt{8}}\frac{\ell}{N}\right)\ktworsc{\ell}{k-1} + \left(\frac{1}{2}\right)^{10N^{1/8}} \right) \\
          &\leq \sum_{\ell=0}^{i-1} \sqrt{2}\left(\sqrt{\frac{\mu_3(\ell)}{N}} + \new{\sqrt{8}}\frac{\ell}{N}\right)\ktworsc{\ell}{k-1} +  \sqrt{2}\cdot 2^{-10N^{1/8}} \cdot N^{1/2} \\
          &\leq \sum_{\ell=0}^{i-1} \sqrt{2}\left(\sqrt{\frac{\mu_3(\ell)}{N}} + \new{\sqrt{8}}\frac{\ell}{N}\right)\ktworsc{\ell}{k-1} +  \sqrt{2}\cdot 2^{-9.5N^{1/8}},
\end{align*}
where the second inequality comes from the recursion on the first term $\ktworsc{i-1}{k}$, and using the fact that $\ktworsc{0}{k} = 0$.
For the third inequality, we used Lemma \ref{lemma:j2col} and the definition of $\mu_3$.
Expanding recursively inside the sum, we have:
\begin{align*}
  \ktworsc{i}{k} &\leq \sum_{\ell=0}^{i-1} \sqrt{2}\left(\sqrt{\frac{\mu_3(\ell)}{N}} + \new{\sqrt{8}}\frac{\ell}{N}\right)\ktworsc{\ell}{k-1} +  \sqrt{2}\cdot 2^{-9.5N^{1/8}} \\
          &\leq \sum_{\ell_1=0}^{i-1} \sqrt{2}\left(\sqrt{\frac{\mu_3(\ell_1)}{N}} + \new{\sqrt{8}}\frac{\ell_1}{N}\right)\Bigg(\sum_{\ell_2=0}^{\ell_1} \sqrt{2}\left(\sqrt{\frac{\mu_3(\ell_2)}{N}} + \new{\sqrt{8}}\frac{\ell_2}{N}\right)\ktworsc{\ell_2}{k-2} \\
          &\quad + \sqrt{2}\cdot  2^{-9.5N^{1/8}}\Bigg)
          +  \sqrt{2}\cdot 2^{-9.5N^{1/8}} \\
          &\vdots \\
          &\leq \sum_{0 \leq \ell_k < \ell_{k-1} <\dots<\ell_1<i} \prod_{j=1}^k\sqrt{2}\left(\sqrt{\frac{\mu_3(\ell_j)}{N}} + \new{\sqrt{8}}\frac{\ell_j}{N}\right) \\
          &\quad + \sqrt{2}\cdot 2^{-9.5N^{1/8}} \sum_{t=0}^{k-1} \sum_{0 \leq \ell_t < \ell_{t-1} <\dots<\ell_1<i} \prod_{j=1}^t\sqrt{2}\left(\sqrt{\frac{\mu_3(\ell_j)}{N}} + \new{\sqrt{8}}\frac{\ell_j}{N}\right) \\
  &\leq \frac{A_i^k}{k!} + \sqrt{2}\cdot 2^{9.5N^{1/8}}\sum_{t=0}^{k-1}\frac{A_i^t}{t!} \\
  &\leq \frac{A_i^k}{k!} + \sqrt{2}\cdot e^{A_i}2^{9.5N^{1/8}},
\end{align*}
where the third inequality comes from expanding recursively all of the terms $\ktworsc{\ell_t}{k-t}$, and using the fact that $\ktworsc{\ell}{0} = 1$.
The fourth inequality comes from the fact that:

\begin{align*}
  &\sum_{0 \leq \ell_k < \ell_{k-1} <\dots<\ell_1<i} \prod_{j=1}^k\sqrt{2}\left(\sqrt{\frac{\mu_3(\ell_j)}{N}} + \new{\sqrt{8}}\frac{\ell_j}{N}\right) \\
  &\leq  \frac{1}{k!} \sum_{0 \leq \ell_k, \ell_{k-1},\dots,\ell_1<i} \prod_{j=1}^k\sqrt{2}\left(\sqrt{\frac{\mu_3(\ell_j)}{N}} + \new{\sqrt{8}}\frac{\ell_j}{N}\right) \\
                                                                                                                      &= \frac{1}{k!} \prod_{j=1}^k \sum_{0 \leq \ell_j <i} \sqrt{2}\left(\sqrt{\frac{\mu_3(\ell_j)}{N}} + \new{\sqrt{8}}\frac{\ell_j}{N}\right) \\
                                                                                                                      &= \frac{1}{k!} \prod_{j=1}^k A_i \\
                                                                                                                      &= \frac{A_i^k}{k!}.
\end{align*}

\subsection{Proof of Lemma~\ref{lemma:Aisbound}}
\label{sec:Aisbound}
\begin{proof}
  We have that:
  \begin{align}
    \nonumber A_{i,s} &= \sum_{\ell=0}^{i-1} \left(\sqrt{(s-1) \cdot \frac{\mu_s(\ell)}{N}} + \new{4} \left(\frac{\ell}{N}\right)^{s/2} + \left(\sum_{r=2}^s \frac{\ell}{N^r}\right)^{1/2}\right)\\
            \label{eq:Ais}&= \sqrt{s-1} \sum_{\ell=0}^{i-1} \sqrt{\frac{\mu_s(\ell)}{N}} + \new{4}\sum_{\ell=0}^{i-1}\left(\frac{\ell}{N}\right)^{s/2} + \sum_{\ell=0}^{i-1}\left(\sum_{r=2}^s \frac{\ell}{N^r}\right)^{1/2}.
  \end{align}

  Notice that
  \begin{align}
    &\nonumber\sum_{\ell=0}^{i-1} \sqrt{\frac{\mu_s(\ell)}{N}} \\ &= \sum_{\ell:\mu_s(\ell) = \new{40} \cdot s^2 \cdot \Pi_{s-1} \cdot N^{1/2^s}} \sqrt{\frac{\new{40} \cdot s^2 \cdot \Pi_{s-1} \cdot N^{1/2^s}}{N}} \\
                                                     & \quad\quad\quad + \sum_{\ell:\mu_s(\ell)>\new{40} \cdot s^2 \cdot \Pi_{s-1} \cdot N^{1/2^s}} \sqrt{\frac{\mu_s(\ell)}{N}} \nonumber \\
                                                     \label{eq:musl}&\leq \sum_{\ell:\mu_s(\ell) = \new{40} \cdot s^2 \cdot \Pi_{s-1} \cdot N^{1/2^s}} \sqrt{\frac{\new{40} \cdot s^2 \cdot \Pi_{s-1} \cdot N^{1/2^s}}{N}} \\
                                                     & \quad\quad\quad+ \sum_{\ell=0}^{i-1} (\new{8}e)^{\frac{2^{s-2}-1}{2^{s-2}}}\frac{\ell^{(2^{s-1}-1)/2^{s-1}}}{N^{(2^{s-2}-1)/2^{s-1}}} \cdot N^{-1/2} \cdot \sqrt{\Pi_{s-1}},\nonumber
  \end{align}
  where we replaced $\mu_s(\ell)$ by its value, and the inequality comes from the fact that there cannot be more than $i$ values such that $\mu_s(\ell)>\new{40} \cdot s^2 \cdot \Pi_{s-1} \cdot N^{1/2^s}$.
  The second summation is at most $(\new{8}e)^{\frac{2^{s-2}-1}{2^{s-2}}}\frac{i^{(2^s-1)/2^{s-1}}}{N^{(2^{s-1}-1)/2^{s-1}}} \cdot \sqrt{\Pi_{s-1}}$.

  For the first summation of \Cref{eq:musl}, we need to count the values of $\ell$ such that $\mu_s(l) = \new{40}s^2 \cdot \Pi_{s-1}\cdot N^{1/2^s}$.
  By using the definition of $\mu_s(\ell)$, this quantity corresponds to the number of $\ell$ that satisfies:
  \begin{align*}
    &\Pi_{s-1} \cdot (\new{8}e)^{\frac{2^{s-2}-1}{2^{s-3}}}\frac{\ell^{(2^{s-1}-1)/2^{s-2}}}{N^{(2^{s-2}-1)/2^{s-2}}} \leq \new{40} \cdot s^2 \cdot \Pi_{s-1} \cdot N^{1/2^s} \\
    \Leftrightarrow& \ell \leq \left(\frac{\new{40}}{(\new{8}e)^{\frac{2^{s-2}-1}{2^{s-3}}}}\right)^{2^{s-2}/(2^{s-1}-1)}\cdot N^{\left(\frac{1}{2^s} + \frac{2^{s-2}-1}{2^{s-2}}\right) \frac{2^{s-2}}{2^{s-1}-1}}  \cdot s^{\frac{2^{s}}{2^{s-1}-1}}\\
    \Leftrightarrow& \ell \leq O\left(s^{\frac{2^{s}}{2^{s-1}-1}} \cdot N^{\left(\frac{1}{2^s} + \frac{2^{s-2}-1}{2^{s-2}}\right) \frac{2^{s-2}}{2^{s-1}-1}}\right).
  \end{align*}

  Thus the first summation of \Cref{eq:musl} is upper-bounded by:
  \begin{align*}
    &\sum_{\ell:\mu_s(\ell) = 10 \cdot \Pi_{s-1} \cdot N^{1/2^s}} \sqrt{\frac{10 \cdot s^2 \cdot \Pi_{s-1} \cdot N^{1/2^s}}{N}} =\\
    &\sqrt{\frac{10 \cdot s^2 \cdot \Pi_{s-1} \cdot N^{1/2^s}}{N}} \cdot O\left(s^{\frac{2^{s}}{2^{s-1}-1}} \cdot N^{\left(\frac{1}{2^s} + \frac{2^{s-2}-1}{2^{s-2}}\right) \frac{2^{s-2}}{2^{s-1}-1}}\right) \\
                                                                                                                     &\leq O\left(N^{-\frac{1}{2} + \frac{1}{2^{s+1}} + \frac{2^{s-3}}{4(2^{s-1}-1)}} \cdot s^4 \cdot \sqrt{\Pi_{s-1}} \right) \\
                                                                                                                     &\leq O\left(N^{\frac{-2^{2s-1}+2^s+2^{s-1}-1+2^{2s-4}}{2(2^{s}-2)}} \cdot s^4 \cdot \sqrt{\Pi_{s-1}} \right) \\
                                                                                                                     &\leq O\left(N^{-1/(2^s(2^s-2))} \cdot s^4 \cdot \sqrt{\Pi_{s-1}} \right) = O\left(N^{-1/(2^s(2^s-2))} \cdot s^4 \cdot \Pi_{s} \right),
  \end{align*}
  where for the first inequality we use that $\frac{2^{s}}{2^{s-1}-1} + 1 \leq 4$ for all $s\geq3$.

  Therefore, we have that:
  \begin{align}
    \sum_{\ell=0}^{i-1} \sqrt{\frac{\mu_s(\ell)}{N}} %
    \label{eq:Ais1}&\leq (2e)^{\frac{2^{s-2}-1}{2^{s-2}}}\frac{i^{(2^s-1)/2^{s-1}}}{N^{(2^{s-1}-1)/2^{s-1}}} \sqrt{\Pi_{s-1}} + O\left(N^{-1/(2^s(2^s-2))} \cdot s^4 \cdot \Pi_s \right).
  \end{align}

  For the second term of \Cref{eq:Ais}, we have:
  \begin{align}
    \nonumber\new{4}\sum_{\ell=0}^{i-1}\left(\frac{\ell}{N}\right)^{s/2} &\leq \new{4}\sum_{\ell=0}^{i-1}  \left(\frac{\ell}{N}\right) \\
    \label{eq:Ais2}&\leq \new{4}\sum_{\ell=0}^{i-1}  \left(\frac{\ell}{N}\right)^{(2^{s-1}-1)/2^{s-1}},
    \end{align}
  where we use that $s\geq3$ and $1\geq (2^{s-1}-1)/2^{s-1}$.
  And for the third term,
  \begin{align}
    \nonumber\sum_{\ell=0}^{i-1} \left(\sum_{r=2}^s \frac{\ell}{N^r}\right)^{1/2} &\leq \sum_{\ell=0}^{i-1} \left((s-1)\frac{\ell}{N^2}\right)^{1/2} \\
    \nonumber&\leq \sum_{\ell=0}^{i-1} \left(\sqrt{s-1} \frac{\ell}{N}\right)\\
    \label{eq:Ais3}&\leq \sum_{\ell=0}^{i-1} \left(\sqrt{s-1}  \left(\frac{\ell}{N}\right)^{(2^{s-1}-1)/2^{s-1}}\right),
  \end{align}
  where we used that $r\geq2$ for the second inequality, and that $1\geq (2^{s-1}-1)/2^{s-1}$ for the last inequality.
  Thus,
  using that \new{$2 \cdot (8e)^{\frac{2^{s-2}-1}{2^{s-2}}} \geq 6$} and combining \Cref{eq:Ais}, \Cref{eq:Ais1}, \Cref{eq:Ais2} and \Cref{eq:Ais3} yields that:
  \begin{align*}
  A_{i,s} &\leq \new{2} \cdot (\new{8}e)^{\frac{2^{s-2}-1}{2^{s-2}}}\frac{i^{(2^s-1)/2^{s-1}}}{N^{(2^{s-1}-1)/2^{s-1}}} \cdot \sqrt{s-1} \cdot \sqrt{\Pi_{s-1}} \\
  & \quad\quad\quad + O\left(N^{-1/(2^s(2^s-2))} \cdot s^4 \cdot  \Pi_s \right)\\
  &= (\new{8}e)^{\frac{2^{s-2} - 1}{2^{s-2}}} \frac{i^{(2^s-1)/2^{s-1}}}{N^{(2^{s-1}-1)/2^{s-1}}} \cdot \Pi_s + O\left( s^4 \cdot \Pi_s \cdot N^{-1/(2^s(2^s-2))}\right).
  \end{align*}
\end{proof}

\subsection{Proof of \Cref{eq:bound-giksrsc-app}}
\label{sec:eq:bound-giksrsc-app}

We have
\begin{align*}
  \srsc{i}{k} \leq& \left(\sum_{\ell=0}^{i-1}B_{\ell,s} \cdot \srsc{\ell}{k-1}\right) + s^{3/2} \cdot 2^{-9.5  \new{\cdot 4} \cdot (s+1)^2 \cdot \Pi_s \cdot N^{1/2^{s+1}}} \\
  \leq&
        \left(\sum_{\ell_1=0}^{i-1}B_{\ell_1,s}\left(\sum_{\ell_2=\ell_1}^{i-1}B_{\ell_2,s} \cdot \srsc{\ell_2}{k-1} + s^{3/2} \cdot 2^{-9.5  \new{\cdot 4} \cdot (s+1)^2 \cdot \Pi_s \cdot N^{1/2^{s+1}}} \right)\right)
        \\
        & \quad\quad\quad + s^{3/2} \cdot 2^{-9.5  \new{\cdot 4} \cdot (s+1)^2 \cdot \Pi_s \cdot N^{1/2^{s+1}}}. \\
\end{align*}
We get by induction
\begin{align*}
  \srsc{i}{k} \leq&
        \Bigg(\sum_{\ell_1=0}^{i-1}B_{\ell_1,s}\Bigg(\sum_{\ell_2=\ell_1}^{i-1}B_{\ell_2,s}\Bigg(\sum_{\ell_3=\ell_2}^{i-1}B_{\ell_3,s} \cdots \\
              & \quad+ s^{3/2} \cdot 2^{-9.5  \new{\cdot 4} \cdot (s+1)^2 \cdot \Pi_s \cdot N^{1/2^{s+1}}} \Bigg) 
                + s^{3/2} \cdot 2^{-9.5  \new{\cdot 4} \cdot (s+1)^2 \cdot \Pi_s \cdot N^{1/2^{s+1}}} \Bigg)\Bigg) \\
                & \quad\quad\quad\quad
                + s^{3/2} \cdot 2^{-9.5  \new{\cdot 4} \cdot (s+1)^2 \cdot \Pi_s \cdot N^{1/2^{s+1}}}.\\
\end{align*}
We thus obtain
\begin{align*}
  \srsc{i}{k} \leq&
     \left(\sum_{0 \leq \ell_k <\ell_{k-1} < \cdots < \ell_1 < i} \prod_{j=1}^k B_{\ell_j,s} \right) \\
     & \quad
     + s^{3/2} \cdot 2^{-9.5  \new{\cdot 4} \cdot (s+1)^2 \cdot \Pi_s \cdot N^{1/2^{s+1}}} \cdot \sum_{t=0}^{k-1} \sum_{0 \leq \ell_t < \ell_{t-1} < \cdots < \ell_1 < i} \prod_{j=1}^t B_{\ell_j,s},\\
\end{align*}
and finally
\begin{align*}
   \srsc{i}{k} \leq& \frac{A_{i,s+1}^k}{k!} + \sum_{\ell=0}^{k-1}\frac{A_{i,s+1}^\ell}{\ell!} \cdot s^{3/2} \cdot 2^{-9.5  \new{\cdot 4} \cdot (s+1)^2 \cdot \Pi_s \cdot N^{1/2^{s+1}}}  \\
  \leq& \frac{A_{i,s+1}^k}{k!} + s^{3/2} \cdot e^{A_{i,s+1}}\cdot 2^{-9.5  \new{\cdot 4} \cdot (s+1)^2 \cdot \Pi_s \cdot N^{1/2^{s+1}}}.
\end{align*}

\subsection{Proof of Lemma~\ref{lemma:algo1kscoptigeneral}}
\label{sec:algo1kscoptigeneral}
\begin{proof}
  Similarly to the proof of Lemma~\ref{lemma:algo1kscgeneral}, we can consider that there are $O\left(N^{k-j+1}\right)$ marked elements in the function $F_1$.
  Hence, using Theorem \ref{theorem:grover}, the second step of the algorithm is expected to make 
  $$O\left(t \cdot \sqrt{\frac{N^k}{N^{k-j+1} \cdot \binom{k}{j}}}\right) = O\left(\frac{t}{\sqrt{\binom{k}{j}}} \cdot N^{(j-1)/2}\right)$$
  quantum queries to the oracle.

  Similarly to the proof of Lemma \ref{lemma:algo1kscgeneral}, we can consider that there are $t \cdot N^{j-1} \cdot \frac{k!}{(j-1)!}$ marked elements in the function $F_2$.
  Hence, using Theorem \ref{theorem:grover}, the fourth step of the algorithm is expected to make 
  $$O\left(\sqrt{\frac{N^k}{t \cdot N^{j-1} \cdot \frac{k!}{(j-1)!}}}\right) = O\left(\frac{N^{(k-j+1)/2}}{\sqrt{t}} \cdot \sqrt{\frac{(j-1)!}{k!}}\right)$$
  quantum queries to the oracle.

  By picking $t=N^{(k-2j+2)/3}$ with $j \leq \frac{k+2}{2}$, the complexity of the algorithm is
  $$O\left(N^{(k-2j+2)/3} \cdot N^{(j-1)/2} \cdot \left(\sqrt{\frac{1}{\binom{k}{j}}} +
      \sqrt{\frac{j!}{k!}}\right)\right) = O\left(\frac{N^{(2k-j+1)/6}}{\sqrt{\binom{k}{j}}}\right)$$
      quantum queries to the oracle.
\end{proof}

\subsection{Proof of Lemma~\ref{lemma:rkscalgogeneral}}
\label{sec:rkscalgogeneral}

\begin{proof}
  We first prove the result when $k$ is divisible by $r+1$.
  The result holds for $r=1$ (using Algorithm \ref{algo:1ksc} and Lemma \ref{lemma:algo1kscgeneral}).

  Fix $r>2$, and assume the result holds for $r-1$.

  The first step of the algorithm is expected to make $O\left( t \cdot N^{k'/2r}\right)$ quantum queries to the oracle if $k'$ is divisible by $r$.

  Similarly to the proof of Lemma~\ref{lemma:algo1kscgeneral}, we can consider that there are $t \cdot N^{k'}$ marked elements in the function $F_1$.
  Hence, using Theorem \ref{theorem:grover}, the third step of the algorithm is expected to make $O\left(\sqrt{\frac{N^k}{t \cdot N^{k'}}}\right)$ quantum queries to the oracle.

  Picking $t=N^{(rk-rk'-k')/3r}$ gives a complexity $O\left(N^{(2rk + (1-2r)k')/6r}\right)$.

  By picking $k'=\frac{r}{r+1}k$, $k'$ is an integer since $k$ is divisible by $r+1$.
  Moreover, $k'$ is divisible by $r$ and the complexity becomes $O\left(N^{k/(2+2r)}\right)$.
  
  If $k$ is not divisible by $r+1$, then there is a $k'$ between $k$ and $k+r+1$ such that $k'$ is divisible by $r+1$.
  Then, we can use Algorithm \ref{algo:rksc} to find a $(r,k')$\---SC with the same functions $h_1,\dots,h_k$ and new random functions $h_{k+1}, \dots, h_{k'}$.
  This gives us a $(r,k)$\---SC for the functions $h_1,\dots,h_k$, and the quantum query complexity is
  $$
  O\left(N^{k'/(2+2r)}\right) \leq O\left(N^{(k+r+1)/(2+2r)}\right).
  $$
\end{proof}

\end{document}